\documentclass[a4paper,11pt,onecolumn,nopdfoutputerror]{quantumarticle}
\usepackage{amsmath,amssymb,amsthm}
\usepackage{booktabs}
\usepackage{microtype}
\usepackage[colorlinks=true,linkcolor=blue!60!black,citecolor=blue!60!black,urlcolor=blue!60!black,hypertexnames=false]{hyperref}

\newtheorem{theorem}{Theorem}
\newtheorem{lemma}{Lemma}
\newtheorem{proposition}[lemma]{Proposition}

\allowdisplaybreaks

\title{The entanglement of purification is not additive}
\author{Artus Krohn-Grimberghe}
\affiliation{Percivio Ltd.}
\date{August 24, 2026}

\begin{document}

\maketitle

\begin{abstract}
The entanglement of purification $E_P$, introduced by Terhal,
Horodecki, Leung, and DiVincenzo in 2002, measures the total
correlations of a bipartite state by the entanglement needed to purify
it; whether it is additive on tensor powers has been open since its
introduction. In 2012 Chen and Winter argued ``beyond reasonable
doubt'', from numerical evidence, that the two-qubit Werner state at
singlet fraction $f = 1/200$ is a counterexample, and explicitly left
a completely rigorous proof to future work. To our knowledge we give
the first rigorous proof: for this state $W$,
$E_P^\infty(W) \le 0.9663\ldots < 97/100 < E_P(W)$, hence the
entanglement of purification is strictly nonadditive at some finite
tensor power. The lower bound is a finite certificate: analytic
reductions send the Schmidt spectrum of every purification in the
complete complex Stiefel domain to a point satisfying explicit
necessary inequalities in a three-dimensional box, and an
exact-rational subdivision of that box has 25{,}383 leaves, each of
which either excludes nonphysical points or proves the entropy bound.
The upper bound combines an exact symmetric-support endpoint, an
explicit feasible isometry, and a convexity theorem for the
regularized quantity, closed by two directed rational logarithm
certificates. Every computer-assisted certificate reduces to integer
comparisons; a dependency-free program verifies all of them from the
supplementary artifact.
\end{abstract}

\section{Introduction and theorem}

How much entanglement does it take to make a given bipartite quantum
state --- not just its entangled part, but all of its correlations?
Terhal, Horodecki, Leung, and DiVincenzo \cite{THLD} answered with the
\emph{entanglement of purification}. For a state $\rho_{AB}$ shared
between two parties, consider every pure state $|\Phi\rangle$ on an
enlarged system $AA'BB'$ whose reduction to $AB$ is $\rho$ (every such
$|\Phi\rangle$ is called a \emph{purification} of $\rho$, with the
first party holding $AA'$ and the second $BB'$). Define
\[
E_P(\rho_{AB}) \;=\; \min_{|\Phi\rangle:\ \mathrm{Tr}_{A'B'}\Phi=\rho}
S(AA')_\Phi ,
\]
the least entanglement entropy of any purification across the
$AA'\,|\,BB'$ cut. Here $S(\tau) = -\mathrm{Tr}\,\tau\log_2\tau$ is
the von Neumann entropy; all logarithms in this paper are base $2$.
The operational meaning proved in \cite{THLD} concerns many copies:
the \emph{regularized} quantity
\[
E_P^\infty(\rho) \;=\; \lim_{n\to\infty}\frac{E_P(\rho^{\otimes n})}{n}
\]
equals the asymptotic entanglement cost of creating $\rho$ by local
operations and a vanishing rate of classical communication.

\begingroup\sloppy
The additivity question asks whether one copy already tells the whole
story: is $E_P(\rho^{\otimes n}) = n\,E_P(\rho)$? Tensoring feasible
purifications shows $E_P(\rho^{\otimes n}) \le n\,E_P(\rho)$ always
(Section~\ref{sec:reduction}), so by Fekete's lemma the limit above
exists and equals $\inf_n E_P(\rho^{\otimes n})/n \le E_P(\rho)$. For
a fixed state $\rho$, equality $E_P^\infty(\rho) = E_P(\rho)$ is
therefore equivalent to additivity along all tensor powers of that
state: if every power is additive the infimum is $E_P(\rho)$;
conversely, if the infimum equals $E_P(\rho)$, then for every $n$ the
infimum definition gives
$n\,E_P(\rho) = n\,E_P^\infty(\rho) \le E_P(\rho^{\otimes n})$, and
subadditivity gives the reverse, so every power is additive. A strict
gap $E_P^\infty(\rho) < E_P(\rho)$ thus means that some finite tensor
power has a strictly smaller per-copy cost.\par\endgroup

Whether such a gap exists has been open since 2002. Bagchi and Pati
\cite{BagchiPati} proved that $E_P$ is subadditive on tensor products
of different states and posed strict subadditivity as an open
question. Chen and Winter \cite{ChenWinter} made the question
concrete: they proved that $E_P^\infty - S$ is convex on ensembles
(their Theorem~3, restated as convexity in their Corollary~4),
computed numerical upper bounds along the two-qubit Werner family, and
found that at singlet fraction $f = 0.005$ the convexity bound forces
$E_P^\infty(W(0.005)) \le 0.9663$ (numerical), while their
local-search numerics for the single-copy value stayed near $0.99$.
They concluded that nonadditivity holds ``beyond reasonable doubt''
--- the phrase is in their title --- but that a gap between numerics
and proof remained, and wrote: ``We leave a completely rigorous proof
of the non-additivity of entanglement of purification to future
work''. The question has stayed open in the literature since: a 2025
status statement describes the additivity of the entanglement of
purification as unknown \cite{Khanian}, and a 2026 paper proving
additivity results for a R\'enyi-2 variant still treats the ordinary
von Neumann Werner case as numerical evidence \cite{FarajiKhanian}.

There is an instructive precedent. For the entanglement of formation,
additivity was eventually refuted by Hastings \cite{Hastings} via
random channel constructions and the equivalence with the additivity
conjectures for channel capacities. No comparable equivalence
machinery is available for $E_P$: it is a minimization over
purifications with growing ancillas rather than a convex-roof over
decompositions, and Hastings's counterexample technique has no known
transfer to it. This is one reason the Werner-state program of Chen
and Winter --- a single concrete two-qubit state with an explicit
numerical target --- remained the most promising route.

This paper completes that program.

\begin{theorem}\label{thm:main}
Let $\Psi_-$ be the projector onto the two-qubit singlet state and
\[
W(f) \;=\; f\,\Psi_- + \frac{1-f}{3}\,(I - \Psi_-),
\qquad f_0 = \frac{1}{200}.
\]
Then
\[
E_P^\infty(W(f_0)) \;\le\; T_+ \;<\; \frac{97}{100} \;<\; E_P(W(f_0)),
\]
where $T_+$ is the exact rational
\[
T_+ =
\frac{966326415938587973999093968736440028761884542960107996118227}
{10^{60}} = 0.96632\ldots
\]
Consequently the entanglement of purification is strictly nonadditive:
for some finite $n$,
\[
E_P\!\left(W(f_0)^{\otimes n}\right) \;<\; n\,E_P(W(f_0)).
\]
The separation is exact, not a rounded decimal:
\[
\frac{97}{100} - T_+ =
\frac{3673584061412026000906031263559971238115457039892003881773}
{10^{60}} \;>\; 0 .
\]
\end{theorem}

Because the proof works through the infimum characterization, it
guarantees only that \emph{some} finite $n$ has a smaller per-copy
value; it identifies neither $n = 2$ nor any upper bound on $n$, and
it produces no exact value of $E_P$ or $E_P^\infty$ at the
nonadditivity parameter $f_0 = 1/200$. (At the endpoint $f = 0$ the
situation is different: Section~\ref{sec:upper} does prove the exact
values $E_P(W(0)^{\otimes n}) = n$ and $E_P^\infty(W(0)) = 1$, with
the trivial-ancilla purification optimal.)

To our knowledge, this is the first rigorous proof of strict
nonadditivity of the ordinary (von Neumann) entanglement of
purification. It completes the concrete Werner-state program of Chen
and Winter: the phenomenon, the state family, the convexity mechanism,
and the parameter $f_0 = 1/200$ are all theirs; the contribution here
is the rigorous certification that converts ``beyond reasonable
doubt'' into a theorem.

Two kinds of argument cooperate in the proof, and we keep their
boundary explicit throughout. Analytic arguments
(Sections~\ref{sec:reduction}, \ref{sec:inequalities},
\ref{sec:upper} and Appendices~A, A$'$) reduce both bounds to finitely
many exact-rational comparisons. Two short computer programs
(Appendix~B) then verify those comparisons; every accepting decision
they make is an exact integer or rational comparison --- many checks
are exact equalities or non-strict inequalities (e.g.\ $U_w \le U_+$),
while the six leaf predicates of Section~\ref{sec:certificate} are
strict. The programs verify the finite calculations; the paper proves
why those calculations apply to every purification.

Notation used throughout: $T_+$ is the rational upper endpoint above;
$U_w$ denotes an exact entropy expression arising from an explicit
purification at $f = 1/100$, and $U_+$ its certified rational upper
bound (Section~\ref{sec:upper}).

\section{Purifications and a scalar reduction}\label{sec:reduction}

\paragraph{The state.}
At $f_0 = 1/200$ the Werner state has exact spectrum
\[
\mathrm{spec}\,W(f_0) =
\left(\frac{1}{200},\ \frac{199}{600},\ \frac{199}{600},\
\frac{199}{600}\right)
\]
in the Bell basis, both single-qubit marginals equal to $I/2$, and
Pauli correlation matrix $t\,I_3$: writing
$\sigma_1,\sigma_2,\sigma_3$ for the Pauli matrices,
\[
\mathrm{Tr}\!\left[W(f_0)\,(\sigma_j\otimes\sigma_k)\right]
= t\,\delta_{jk},
\qquad
t = -f_0 + \frac{1-f_0}{3} = \frac{49}{150}.
\]

\paragraph{The purification domain.}
The minimization defining $E_P$ ranges over purifications with
ancillas $A'$, $B'$ of unbounded dimension. For a state of rank $r$,
Ibinson, Linden, and Winter proved (Corollary~3 of \cite{ILW}) that
some optimal purification uses ancillas of dimension at most $r$ each.
Here $r = 4$, so it suffices to minimize over ancillas
$A' \cong B' \cong \mathbb{C}^4$. Padding any smaller purification
with unused ancilla dimensions changes nothing, so the domain can be
taken to be \emph{exactly} the set of isometries
\[
V:\ \mathbb{C}^4 \longrightarrow \mathbb{C}^4\otimes\mathbb{C}^4,
\qquad V^\dagger V = I_4,
\]
applied to a fixed standard purification of $W(f_0)$: this set of
$16\times 4$ complex matrices with orthonormal columns is the complex
Stiefel manifold $\mathrm{St}_{\mathbb C}(4,16)$. Dimension $4$ is a
\emph{sufficient} padded domain here, not a claimed minimal dimension.
The reduction is spelled out in Appendix~A.1.

\paragraph{Schmidt variables.}
Fix any $V \in \mathrm{St}_{\mathbb C}(4,16)$ and write the resulting
pure state in Schmidt form across the cut $AA'\,|\,BB'$: there are
orthonormal families $\{|x_\alpha\rangle\}$ on $AA'$ and
$\{|y_\alpha\rangle\}$ on $BB'$ and numbers
$p_1 \ge p_2 \ge \cdots \ge p_r > 0$ summing to one with
\[
|\Phi_V\rangle = \sum_{\alpha=1}^{r}\sqrt{p_\alpha}\,
|x_\alpha\rangle|y_\alpha\rangle .
\]
Both sides of the cut have dimension $8$, so $r \le 8$; pad the
spectrum with zeros to a fixed-length vector $p = (p_1, \ldots, p_8)$,
$p_1 \ge \cdots \ge p_8 \ge 0$, so that entries such as $p_2, p_3$ are
defined for every purification regardless of its Schmidt rank. The
objective is exactly the Shannon entropy of the Schmidt spectrum:
$S(AA')_{\Phi_V} = H(p) = \sum_\alpha h(p_\alpha)$ with
$h(x) = -x\log_2 x$ and $h(0) = 0$. No restriction is imposed on $V$
within the domain: it is an arbitrary complex isometry, and ordering
the $p_\alpha$ merely fixes labels. All spectral inequalities proved
in Section~\ref{sec:inequalities} are statements about purifications
in this sufficient compact domain $\mathrm{St}_{\mathbb C}(4,16)$; by
the Appendix~A.1 reduction, the minimum of $H(p)$ over this domain
equals $E_P(W(f_0))$, which is all the lower bound needs.

\paragraph{Roadmap.}
The whole lower-bound argument is a pincer on the ordered spectrum.
Section~\ref{sec:inequalities} proves that \emph{every} spectrum
arising from a physical purification satisfies four explicit
inequalities in the variables
\[
p_1,\quad p_2,\quad p_3,\quad Q = \sum_{\alpha\ge4}p_\alpha
= 1 - p_1 - p_2 - p_3 .
\]
Section~\ref{sec:certificate} then certifies, by an exact finite
subdivision of the three-dimensional box
$[1/2,1]\times[0,1/2]\times[0,1/3]$ in the coordinates
$(p_1,p_2,p_3)$, that every point compatible with those inequalities
has $H > 97/100$. We emphasize the division of labor: the subdivision
lives in a three-dimensional box of \emph{spectra}, reached by
analytic reduction; it does not partition or parameterize the Stiefel
manifold itself. Section~\ref{sec:upper} proves the upper bound on
$E_P^\infty$ by an explicit protocol chain, and
Section~\ref{sec:conclusion} collects scope and context.

\section{Necessary spectral inequalities}\label{sec:inequalities}

This section states the inequalities and displays the mechanism of
each proof; complete proofs are in Appendix~A. Throughout, a
\emph{qubit channel} is a completely positive trace-preserving linear
map on qubit states, and its \emph{Bloch form} is the affine map
$r \mapsto c + Lr$ it induces on Bloch vectors, where a qubit state
with Bloch vector $r\in\mathbb{R}^3$, $\|r\|\le1$, is
$(I + r\cdot\sigma)/2$. We write $s = p_1 + p_2$,
$\varepsilon = 1 - s$, and $z = p_3$.

\subsection{The orientation inequality}

\begin{lemma}[orientation cap]\label{lem:orientation}
If $L$ is the $3\times3$ linear Bloch part of a qubit channel, then
$|\det L| \le 1$; and if $\det L < 0$, then $|\det L| \le 1/27$.
\end{lemma}

The mechanism: polar-decompose $L = OP$ with $O$ improper orthogonal
and $P > 0$, average the channel over rotations (which removes the
translation $c$ and isotropizes $P$ to $\alpha I$ with
$\alpha = \mathrm{Tr}\,P/3$), and read off from the four Pauli
probabilities of the averaged channel that complete positivity forces
$\alpha \le 1/3$; the arithmetic--geometric mean inequality then gives
$\det P \le \alpha^3 \le 1/27$. The full calculation, including why
the post-rotation stays proper and both determinant cases, is
Appendix~A.4.

To use Lemma~\ref{lem:orientation} we must locate qubit channels
inside an arbitrary purification. The two leading Schmidt vectors on
each side span a two-dimensional ``seed'' space; compressing the
physical qubit's Pauli operators to the Schmidt bases converts the
exact correlation identity
$\mathrm{Tr}[W(\sigma_j\otimes\sigma_k)] = t\delta_{jk}$ into the
matrix decomposition
\[
t\,I_3 \;=\; s\,A\,C_\delta\,B^{T} \;+\; s\,a b^{T} \;+\; R ,
\]
whose ingredients are: $A$ and $B$, the linear Bloch parts of two
qubit channels built from the leading Schmidt vectors of the two
sides; the diagonal matrix $C_\delta$ with diagonal entries
$\sqrt{1-\delta^2}$, $-\sqrt{1-\delta^2}$, $1-\delta^2$, where
$\delta = (p_1-p_2)/s$, which records how unevenly
the two leading Schmidt weights are split; the vectors $a,b$, the
output Bloch vectors of those channels on the normalized leading-two
seed state; and a remainder $R$ collecting every term that touches
Schmidt index $3$ or higher. Appendix~A.5 constructs the channels and
proves the identity.

\begin{lemma}[remainder bound]\label{lem:remainder}
In the decomposition above,
\[
\|R\|_{\mathrm{op}} \le \varepsilon + 2sz,
\qquad
\|a\|,\ \|b\| \le \frac{\varepsilon}{s},
\]
so with $D = sAC_\delta B^T$,
\[
\|D - tI_3\|_{\mathrm{op}} \;\le\;
\eta := \varepsilon + 2sz + \frac{\varepsilon^2}{s}.
\]
\end{lemma}

The three contributions are proved in two appendices. Appendix~A.3
proves the two remainder bounds: the cross terms between the leading
pair and the tail are bounded by $2sz$ through a shared row budget in
Bessel's inequality (the appearance of the \emph{largest tail
probability} $z$, rather than $\sqrt\varepsilon$, is what makes the
bound strong enough), and the tail--tail terms are an expectation of a
norm-one observable in the normalized tail state, hence at most
$\varepsilon$. Appendix~A.5 proves the marginal-cancellation bound:
the marginal conditions force the seed output Bloch vectors to cancel
against a tail of weight $\varepsilon$, giving
$\|a\|,\|b\| \le \varepsilon/s$ and hence
$\|sab^T\|_{\mathrm{op}} \le \varepsilon^2/s$.

\begin{proposition}[orientation polynomial]\label{prop:orientation}
Every physical ordered spectrum satisfies
\begin{equation}
G(p_1,p_2,p_3) \;:=\; 16\,p_1^2\,p_2^2\,s^2 \;-\; 27\,N^3 \;\ge\; 0,
\tag{1}\label{eq:G}
\end{equation}
where $N = s(t - \varepsilon - 2sz) - \varepsilon^2 = s(t-\eta)$.
\end{proposition}

Mechanism: if $t > \eta$, Lemma~\ref{lem:remainder} keeps the segment
from $tI_3$ to $D$ away from singular matrices, so
$\det D \ge (t-\eta)^3 > 0$; but
$\det D = -s^3(1-\delta^2)^2\det A\det B$, so exactly one of the two
channels reverses orientation, and Lemma~\ref{lem:orientation} caps
$\det D \le s^3(1-\delta^2)^2/27$. Substituting
$1-\delta^2 = 4p_1p_2/s^2$ and clearing positive denominators gives
\eqref{eq:G}. If instead $t \le \eta$, then $N \le 0$ and \eqref{eq:G}
holds trivially --- so \eqref{eq:G} is necessary on the entire domain,
with no hidden case assumption (Appendix~A.5).

\subsection{The spectral square-root inequality}

The second inequality comes from a different physical mechanism:
symmetric extendibility. A state $\sigma_{XY}$ has a \emph{symmetric
extension} (in $Y$) if there is a state $\sigma_{XYY'}$ with $Y'$ a
copy of $Y$, invariant under swapping $Y$ and $Y'$, whose reduction to
$XY$ is $\sigma_{XY}$.

\begin{lemma}[symmetric-extension spectral lemma]\label{lem:spectral}
Let $\sigma_{XY}$ be a state with $Y$ a qubit and $\sigma_Y = I/2$
that admits a symmetric extension in $Y$. If
$q_1 \ge q_2 \ge q_3 \ge \cdots$ are its eigenvalues and
$Q_\sigma = \sum_{i\ge4} q_i$, then
\[
\sqrt{q_1} \;\le\; \sqrt{q_2} + \sqrt{q_3} + \sqrt{Q_\sigma}.
\]
\end{lemma}

The dimension of $X$ is arbitrary. The proof (Appendix~A.6)
compresses $X$ to the two-dimensional Schmidt support of a leading
eigenvector, which preserves symmetric extendibility, and applies the
two-qubit symmetric-extension criterion of Chen, Ji, Kribs,
L\"utkenhaus, and Zeng \cite{CJKLZ}: a two-qubit state $\omega$ has a
symmetric extension in $Y$ if and only if
$\mathrm{Tr}\,\omega_Y^2 \ge \mathrm{Tr}\,\omega^2 -
4\sqrt{\det\omega}$. After careful bookkeeping (eight distinct steps,
each spelled out in Appendix~A.6), the criterion factors into the
``Heron-like'' identity
\[
(a{+}b{+}c{-}d)(a{+}b{-}c{+}d)(a{-}b{+}c{+}d)({-}a{+}b{+}c{+}d)
\ \ge\ 0
\]
in the variables $a = \sqrt{q_1}$, $b = \sqrt{q_2}$,
$c = \sqrt{q_3}$, $d = \sqrt{Q_\sigma}$; the first three factors are
positive, so the fourth is nonnegative, which is the claim.

To apply Lemma~\ref{lem:spectral}, split the $8\times8$ coefficient
matrix $M$ of the purification across $AA'\,|\,BB'$ into its physical
triplet and singlet parts, $M = S + E$; the singlet weight of the
Werner state gives exactly $\|E\|_F^2 = f_0 = 1/200$. The normalized
triplet part is swap-invariant on the physical qubits, so its $A'A$
reduction has a symmetric extension and a maximally mixed qubit
marginal --- precisely the hypotheses of Lemma~\ref{lem:spectral}. Its
eigenvalues are the normalized squared singular values of $S$, and
Mirsky's singular-value perturbation inequality \cite{Mirsky}
transfers the conclusion from $S$ back to $M$ at the cost of a penalty
$2\sqrt{1/200} = 1/\sqrt{50}$:

\begin{proposition}[spectral inequality]\label{prop:spectral}
Every physical ordered spectrum satisfies
\begin{equation}
\sqrt{p_1} \;\le\; \sqrt{p_2} + \sqrt{p_3} + \sqrt{Q}
+ \frac{1}{\sqrt{50}} .
\tag{2}\label{eq:spectral}
\end{equation}
\end{proposition}

\subsection{Tail cap and entropy merge}

Two elementary facts complete the list. Since $r \le 8$ and the
spectrum is ordered, the tail $Q$ is a sum of at most five terms each
at most $p_3$:
\begin{equation}
0 \;\le\; Q \;\le\; 5\,p_3 .
\tag{3}\label{eq:tailcap}
\end{equation}
And since $h(x+y) \le h(x)+h(y)$ for $x,y \ge 0$ with $x+y\le1$
(concavity of $h$ with $h(0)=0$), merging the tail into a single mass
only decreases entropy:
\begin{equation}
H(p) \;\ge\; h(p_1) + h(p_2) + h(p_3) + h(Q) .
\tag{4}\label{eq:entropymerge}
\end{equation}
Finally, if $p_1 \le 1/2$, the ordered spectrum is majorized by
$(1/2, 1/2, 0, \ldots)$, and Schur concavity of entropy gives
$H(p) \ge 1 > 97/100$ directly. The certified subdivision therefore
only needs to handle $p_1 \in [1/2, 1]$.

\section{The certified lower bound}\label{sec:certificate}

\paragraph{The cover.}
Consider the closed box
\[
\mathcal B = [1/2,\,1]\times[0,\,1/2]\times[0,\,1/3]
\]
in the coordinates $(p_1,p_2,p_3)$, with $Q = 1-p_1-p_2-p_3$ treated
as a derived quantity. By Section~\ref{sec:inequalities}, the ordered
Schmidt spectrum of every purification with $p_1 \ge 1/2$ yields a
point of $\mathcal B$ (the ranges of $p_2$, $p_3$ follow from ordering
and normalization) satisfying $Q \ge 0$, $p_2 \ge p_3$, and
inequalities \eqref{eq:G}--\eqref{eq:tailcap}.

A deterministic binary subdivision refines $\mathcal B$: each box is
split at the exact rational midpoint of its widest coordinate, with
the first coordinate winning ties. Subdivision stops at a box as soon
as one of six \emph{strict} interval predicates certifies it. Writing
$[x_{\mathrm{lo}}, x_{\mathrm{hi}}]$ for the exact rational range of a
quantity $x$ over the box (computed by outward interval arithmetic:
every intermediate result is enclosed by rationals rounded away from
the true value), the predicates and their meanings are:

\begin{center}
\small
\begin{tabular}{@{}lrlp{0.26\textwidth}@{}}
\toprule
Class & Leaves & Certified condition & Meaning \\
\midrule
\texttt{qneg} & 420 & $Q_{\mathrm{hi}} < 0$ &
no point of the box has a valid tail weight $Q \ge 0$ \\
\texttt{order23} & 184 & $(p_2)_{\mathrm{hi}} < (p_3)_{\mathrm{lo}}$ &
no point respects the ordering $p_2 \ge p_3$ \\
\texttt{tailcap} & 107 & $Q_{\mathrm{lo}} > 5\,(p_3)_{\mathrm{hi}}$ &
every point violates \eqref{eq:tailcap} \\
\texttt{orientation} & 3810 & $G_{\mathrm{hi}} < 0$ &
every point violates \eqref{eq:G} \\
\texttt{spectral} & 10066 &
$\bigl(\sqrt{p_2}{+}\sqrt{p_3}{+}\sqrt{Q}{+}\tfrac{1}{\sqrt{50}}
\bigr)_{\mathrm{hi}} < (\sqrt{p_1})_{\mathrm{lo}}$ &
every point violates \eqref{eq:spectral} \\
\texttt{entropy} & 10796 &
$\bigl(h(p_1){+}h(p_2){+}h(p_3){+}h(Q)\bigr)_{\mathrm{lo}} >
\tfrac{97}{100}$ &
\eqref{eq:entropymerge} forces $H > 97/100$ on the whole box \\
\bottomrule
\end{tabular}
\end{center}

The six classes are not literally one predicate per inequality: the
first two are domain and ordering exclusions, the middle three certify
violations of \eqref{eq:tailcap}, \eqref{eq:G}, \eqref{eq:spectral}
respectively, and the last combines the entropy merge
\eqref{eq:entropymerge} with the target threshold. The subdivision
terminates with exactly $25{,}383$ leaves
($420 + 184 + 107 + 3810 + 10066 + 10796$), reached through $50{,}765$
tree nodes with maximum depth $38$.

\paragraph{Completeness.}
The leaf paths (strings of left/right choices from the root) are
prefix-free, every internal node has both children, and the Kraft sum
$\sum_{\text{leaves}} 2^{-\mathrm{depth}}$ is exactly $1$; together
these facts imply that the closed leaf boxes cover all of
$\mathcal B$. Points on shared boundaries may lie in two closed
leaves; this is harmless because every predicate holds at every point
of its closed box.

\paragraph{A worked leaf.}
Follow the path left--right--right from the root. The root's widest
coordinates are $p_1$ and $p_2$ (width $1/2$ each); $p_1$ wins the
tie, and the left child keeps $p_1 \in [1/2, 3/4]$. Now $p_2$ (width
$1/2$) is widest; the right child keeps $p_2 \in [1/4, 1/2]$. Now
$p_3$ (width $1/3$) is widest; the right child keeps
$p_3 \in [1/6, 1/3]$. The resulting box
\[
[1/2,\,3/4]\times[1/4,\,1/2]\times[1/6,\,1/3]
\]
is a depth-3 \texttt{entropy} leaf. Indeed $Q \in [-7/12,\,1/12]$,
clipped to $[0, 1/12]$ for the entropy term; since $h$ is concave, its
minimum over an interval is at an endpoint, so the four contributions
are at least
\[
\min\{h(\tfrac12), h(\tfrac34)\} + \min\{h(\tfrac14), h(\tfrac12)\}
+ \min\{h(\tfrac16), h(\tfrac13)\} + \min\{h(0), h(\tfrac1{12})\},
\]
\begingroup\sloppy
which the certificate bounds below by certified rational enclosures;
numerically the four minima are $0.3112\ldots$, $0.5$ (exactly, both
endpoints), $0.4308\ldots$, and $0$, summing to
$1.242\ldots > 97/100$. For contrast, the sibling path
right--right--right produces the box
$[3/4,1]\times[1/4,1/2]\times[1/6,1/3]$, a depth-3 \texttt{qneg}
leaf: there
$Q_{\mathrm{hi}} = 1 - \tfrac34 - \tfrac14 - \tfrac16 = -\tfrac16 <
0$, so no spectrum lands in it at all.\par\endgroup

\paragraph{Conclusion of the lower bound.}
Let $p$ be the (zero-padded) ordered Schmidt spectrum of any
purification in the domain $\mathrm{St}_{\mathbb C}(4,16)$ of
Section~\ref{sec:reduction}. If $p_1 \le 1/2$ then $H(p) \ge 1$. If
$p_1 \ge 1/2$, the point $(p_1,p_2,p_3)$ lies in some closed leaf of
the cover; it cannot lie in a leaf of the first five classes, because
Section~\ref{sec:inequalities} shows physical spectra satisfy
$Q\ge0$, the ordering, and \eqref{eq:G}--\eqref{eq:tailcap}; hence it
lies in an \texttt{entropy} leaf, and \eqref{eq:entropymerge} gives
$H(p) > 97/100$. The minimum of $H(p)$ over this compact domain is
attained, and by the Appendix~A.1 reduction that minimum equals
$E_P(W(f_0))$ --- the minimum over \emph{all} purifications with
arbitrary ancillas; hence
\begin{equation}
E_P\!\left(W(1/200)\right) \;>\; \frac{97}{100}.
\tag{5}\label{eq:lower}
\end{equation}

Appendix~B specifies the exact certificate semantics --- outward
rational intervals, dyadic square-root enclosures checked by integer
squaring, directed logarithm series with explicit remainder bounds ---
and the verification program that reconstructs the entire tree from
the split rule alone and rechecks every leaf at higher precision.

\section{The regularized upper bound}\label{sec:upper}

The upper bound is a chain of three exact ingredients and one
convexity theorem.

\paragraph{The endpoint $f = 0$.}
$W(0) = (I-\Psi_-)/3 = P_{\mathrm{sym}}/3$ is the normalized projector
onto the symmetric (triplet) subspace. For states supported on the
symmetric subspace, purifications can be chosen invariant under
swapping the two physical qubits, and a weak monotonicity argument
(Appendix~A.9; the general statement is Proposition~7 of Christandl
and Winter \cite{CW}) pins the value exactly:
\[
E_P\!\left(W(0)^{\otimes n}\right) = n
\quad\text{for every } n,
\qquad\text{hence}\qquad
E_P^\infty(W(0)) = 1 .
\]

\paragraph{The endpoint $f = 1/100$.}
$W(1/100)$ has Bell-basis spectrum
$(1/100,\allowbreak\,33/100,\allowbreak\,33/100,\allowbreak\,
33/100)$. Fix the Bell convention
\begin{gather*}
|\Psi_0\rangle = \tfrac{1}{\sqrt2}(|01\rangle - |10\rangle),\qquad
|\Psi_1\rangle = \tfrac{1}{\sqrt2}(|00\rangle + |11\rangle),\\
|\Psi_2\rangle = \tfrac{1}{\sqrt2}(|00\rangle - |11\rangle),\qquad
|\Psi_3\rangle = \tfrac{1}{\sqrt2}(|01\rangle + |10\rangle),
\end{gather*}
and let $\lambda_0 = 1/100$,
$\lambda_1=\lambda_2=\lambda_3 = 33/100$. Starting from the standard
purification
$\sum_i \sqrt{\lambda_i}\,|\Psi_i\rangle_{AB}|i\rangle_{A'}$, apply
the isometry $V: \mathbb{C}^4 \to A'B'$ defined by
$V|i\rangle = |\Psi_i\rangle_{A'B'}$ --- each purifying label goes to
the corresponding ancillary Bell state. Its columns are orthonormal,
so this is a feasible purification:
\[
|\psi\rangle \;=\; \sum_{i=0}^{3}\sqrt{\lambda_i}\,
|\Psi_i\rangle_{AB}\,|\Psi_i\rangle_{A'B'} .
\]
The ancillary Bell states are two-qubit vectors; in the
$\mathbb{C}^4\otimes\mathbb{C}^4$ ancilla framework of
Section~\ref{sec:reduction} they occupy a
$\mathbb{C}^2\otimes\mathbb{C}^2$ subspace (the first two dimensions
of each ancilla factor), so the $8\times 8$ marginal on $AA'$ has the
block form $\rho_{4\times4} \oplus 0_4$: a rank-four $4\times4$ block
plus four exact zero eigenvalues. Its nonzero spectrum is
\[
\{\,s,\ s,\ s,\ 1-3s\,\},
\qquad
s = \frac{17-\sqrt{33}}{200},
\]
verified by exact power sums: the verification program confirms
$\mathrm{Tr}\,\rho_{AA'}^k = 3s^k + (1-3s)^k$ for $k = 1,2,3,4$ in
exact arithmetic over $\mathbb{Q}[\sqrt{33}]$, and four exact power
sums determine a multiset of four nonzero eigenvalues (Newton's
identities). This spectrum check is one of the computer-assisted steps
(Appendix~B). Therefore
\[
E_P\!\left(W(1/100)\right) \;\le\; U_w :=
-3s\log_2 s - (1-3s)\log_2(1-3s),
\]
and the first directed rational certificate proves
\begin{equation}
\begin{split}
U_w \;\le\; U_+ &=
\frac{922616583105498918071628209679061634159251672167426749996626}
{10^{60}}\\
&= 0.92261\ldots
\end{split}
\tag{6}\label{eq:Uplus}
\end{equation}
By subadditivity and Fekete's lemma (Section~1),
$E_P^\infty(W(1/100)) \le E_P(W(1/100)) \le U_+$.

\paragraph{The convexity theorem.}
The bridge between the endpoints is the finite-ensemble inequality
first stated by Chen and Winter (Theorem~3 and Corollary~4 of
\cite{ChenWinter}): for every finite ensemble
$\rho = \sum_i p_i\rho_i$,
\begin{equation}
E_P^\infty(\rho) \;\le\; \sum_i p_i\,E_P^\infty(\rho_i)
\;+\; S(\rho) - \sum_i p_i S(\rho_i),
\tag{7}\label{eq:convexity}
\end{equation}
equivalently, $E_P^\infty - S$ is convex on ensembles. Because
\cite{ChenWinter} is an unpublished preprint whose proof of this
theorem is a summary outline, Appendix~A$'$ gives a complete proof of
\eqref{eq:convexity} from published ingredients --- the operational
identity of \cite{THLD}, entanglement dilution with sublinear
communication \cite{LoPopescu}, standard typicality estimates
\cite{WinterCMP, WinterIT}, and the operator-sampling bound of
Ahlswede and Winter \cite{AW}, in the covering form used by Winter
\cite{WinterCMP} and by Groisman, Popescu, and Winter \cite{GPW}.
This paper therefore does not rest on any unpublished proof.

\paragraph{Closing the chain.}
The Werner family is affine in $f$, so exactly
\begin{equation}
W\!\left(\tfrac{1}{200}\right) \;=\;
\tfrac12\,W(0) \;+\; \tfrac12\,W\!\left(\tfrac{1}{100}\right).
\tag{8}\label{eq:mixture}
\end{equation}
Applying \eqref{eq:convexity} to this two-state ensemble and
inserting the endpoint values,
\begin{equation}
E_P^\infty\!\left(W(\tfrac{1}{200})\right)
\;\le\;
S\!\left(W(\tfrac{1}{200})\right)
+ \tfrac12\!\left(1 - S(W(0))\right)
+ \tfrac12\!\left(U_+ - S\!\left(W(\tfrac{1}{100})\right)\right),
\tag{9}\label{eq:chain}
\end{equation}
with the three exact entropies
\begin{gather*}
S(W(0)) = \log_2 3,
\qquad
S\!\left(W(\tfrac{1}{100})\right)
= -\tfrac{1}{100}\log_2\tfrac{1}{100}
  -\tfrac{99}{100}\log_2\tfrac{33}{100},\\
S\!\left(W(\tfrac{1}{200})\right)
= -\tfrac{1}{200}\log_2\tfrac{1}{200}
  -\tfrac{199}{200}\log_2\tfrac{199}{600}.
\end{gather*}
The second directed rational certificate bounds the right-hand side of
\eqref{eq:chain} strictly below $T_+$, using the same outward
logarithm semantics as the lower-bound certificate (Appendix~B).
Therefore
\begin{equation}
E_P^\infty\!\left(W(1/200)\right) \;\le\; T_+ .
\tag{10}\label{eq:upper}
\end{equation}
Combining \eqref{eq:lower}, \eqref{eq:upper}, and the exact positive
margin $97/100 - T_+ > 0$ displayed in Theorem~\ref{thm:main}
completes the proof: $E_P^\infty(W(1/200)) < E_P(W(1/200))$, so some
finite tensor power beats $n$ times the single-copy value.
\hfill$\blacksquare$

\section{Conclusion, limitations, and context}\label{sec:conclusion}

The entanglement of purification is not additive. The counterexample
is exactly the one Chen and Winter proposed: the two-qubit Werner
state at $f = 1/200$. The following table compares their numerical
program with the certified quantities proved here; every number in the
middle column is a numerical (local-search or numerically derived)
value from \cite{ChenWinter} and is labeled as such.

\begin{center}
\footnotesize
\begin{tabular}{@{}lll@{}}
\toprule
Quantity & Chen--Winter \cite{ChenWinter} (numerical) &
This paper (proved) \\
\midrule
$E_P(W(0))$ & $= 1$ & $= 1$ (Appendix~A.9) \\
$E_P(W(1/100))$ upper bound & $\le 0.9226$ (local search) &
$\le U_+ = 0.92261\ldots$ (exact rational) \\
$E_P^\infty(W(1/200))$ upper bound & $\le 0.9663$ (via their Thm.~3) &
$\le T_+ = 0.96632\ldots$ (exact rational) \\
$E_P(W(1/200))$ lower & $\gtrsim 0.99$ suggested, unproved &
$> 97/100$ (certified cover) \\
\bottomrule
\end{tabular}
\end{center}

The agreement in the second row to four decimal places suggests that
their local search had already found (a state equivalent to) the
explicit Bell-to-Bell purification of Section~\ref{sec:upper}; the
present contribution is that both rows are now theorems rather than
numerics, and that the last row --- the step their paper explicitly
left open --- is closed by a finite certificate.

\begin{quote}
\textbf{Scope.} The proof establishes strict nonadditivity at some
finite tensor power $n$. It yields no exact value of $E_P$ or
$E_P^\infty$ at the nonadditivity parameter $f = 1/200$; at the
endpoint $f = 0$ it does prove the exact values
$E_P(W(0)^{\otimes n}) = n$ and $E_P^\infty(W(0)) = 1$, with the
trivial-ancilla purification optimal (Appendix~A.9). The explicit
purification at $f = 1/100$ is a feasible witness, not shown optimal,
and no optimizer at $f = 1/100$ or $f = 1/200$ is identified. No
minimal or explicit violating $n$ is produced. No classification of
orientation-reversing channels is claimed beyond the necessary
determinant inequality \eqref{eq:G}. No separate partial search for
the exact minimum is used in this proof.
\end{quote}

The choice of threshold $97/100$ has slack on both sides: the
certified cover proves $E_P > 0.97$ while the numerics of
\cite{ChenWinter} suggest the truth is near $0.99$, and the upper
chain gives $T_+ \approx 0.9663$. Any threshold strictly between $T_+$
and the true single-copy value would do; $97/100$ keeps the
certificate small.

Two contrasts situate the result. First, for the R\'enyi-2 variant of
the entanglement of purification, additivity \emph{holds} on relevant
classes \cite{FarajiKhanian} --- the nonadditivity proved here is a
genuinely von Neumann phenomenon on this family, and the contrast
sharpens the question of which entropies support additive purification
measures. Second, in holography the entanglement of purification is
conjecturally identified with the entanglement wedge cross-section
\cite{TU}; that identification concerns a different (large-$N$,
geometric) regime, and nothing here bears on it beyond removing the
possibility that additivity holds as an abstract theorem.

Beyond the specific open problem it closes, the paper illustrates a
transferable proof pattern: replace an intractable optimization over a
continuous manifold by \emph{necessary} scalar inequalities on a
low-dimensional summary of the optimizer (here, four Schmidt-spectrum
inequalities), then certify the resulting finite-dimensional claim by
an exact-rational subdivision whose every accepting comparison is an
integer comparison. The pattern needs no floating-point trust anywhere
and produces certificates a short independent program can check.

\section*{Acknowledgments}

The author thanks the authors of \cite{ChenWinter} for posing the
problem in concrete, checkable form.

\section*{Contribution statement}

The sole author takes full responsibility for the entire content.
Scope of AI use in producing this work: AI systems were used for code
production, for calculations, for the mechanical derivation of proof
steps, and for drafting text under the author's direction. The
verification chain is independent of how the derivations were
produced: every computational claim reduces to big-integer comparisons
checkable by the supplied dependency-free verifier, and every imported
mathematical statement is either proved in the appendices or cited to
the published literature.

\section*{Data and code availability}

The supplementary artifact contains two dependency-free Python
programs (standard library only), a manifest of SHA-256 hashes, and
run instructions. \texttt{verify\_leaves.py} reconstructs the entire
25{,}383-leaf subdivision from the split rule alone and rechecks every
leaf predicate at increased precision;
\texttt{verify\_upper\_chain.py} checks the endpoint isometry, the
exact spectrum power sums, both directed rational logarithm
certificates, the exact mixture identity, and the final margin. Both
print \texttt{PASS} and exit successfully only if every check
succeeds; both reject deliberately corrupted test variants of the
certificate. The programs verify the finite calculations; the paper
proves why those calculations apply to every purification.
The artifact is archived at
DOI~\href{https://doi.org/10.5281/zenodo.22097511}{10.5281/zenodo.22097511}.

\appendix

\section{Complete analytic proofs}

Throughout, $|\Phi_V\rangle$ is the purification determined by
$V \in \mathrm{St}_{\mathbb C}(4,16)$ as in
Section~\ref{sec:reduction}, with ordered Schmidt data
$(p_\alpha, |x_\alpha\rangle, |y_\alpha\rangle)$, $r \le 8$, and
$s = p_1+p_2$, $\varepsilon = 1-s$, $z = p_3$,
$Q = \sum_{\alpha\ge4}p_\alpha$.

\subsection{Ancilla reduction and symbol map}

The definition of $E_P$ minimizes $S(AA')$ over purifications with
arbitrary finite ancillas. Corollary~3 of Ibinson--Linden--Winter
\cite{ILW} states that for any bipartite state $\rho_{AB}$ there is an
optimal purification in which each ancillary factor has dimension at
most $\mathrm{rank}\,\rho$; for $W(f_0)$, $\mathrm{rank} = 4$. Fix the
standard purification
\[
|\varphi\rangle = \sum_{i=0}^{3}\sqrt{\lambda_i}\,
|\Psi_i\rangle_{AB}\,|i\rangle_{A'},
\]
with $\lambda_i$ the Bell-basis eigenvalues of $W(f_0)$ and the Bell
convention of Section~\ref{sec:upper}. Every purification with
ancillas $A'B' \cong \mathbb{C}^4\otimes\mathbb{C}^4$ is
$(I_{AB}\otimes V)|\varphi\rangle$ for exactly one isometry
$V:\mathbb{C}^4 \to \mathbb{C}^4\otimes\mathbb{C}^4$ (this is the
standard fact that all purifications of a state differ by an isometry
on the purifying system). A purification with smaller ancillas embeds
into this family by padding, and padding does not change $S(AA')$.
Hence
\[
E_P(W(f_0)) = \min_{V\in\mathrm{St}_{\mathbb C}(4,16)}
S(AA')_{\Phi_V},
\]
and the minimum is attained because the Stiefel manifold is compact
and $V \mapsto S(AA')$ is continuous. Dimension $4$ is sufficient by
\cite{ILW}; no claim of necessity is made.

\subsection{Correlation equations for Schmidt data}

Define isometries $U_A|\alpha\rangle = |x_\alpha\rangle$ and
$U_B|\alpha\rangle = |y_\alpha\rangle$ from an $r$-dimensional index
space into $AA'$ and $BB'$. For unit vectors $u, v \in \mathbb{R}^3$
set
\[
X(u) = U_A^\dagger\left[(u\cdot\sigma)_A\otimes I_{A'}\right]U_A,
\qquad
Y(v) = U_B^\dagger\left[(v\cdot\sigma)_B\otimes I_{B'}\right]U_B .
\]
These are Hermitian, and being compressions of unitaries they are
contractions with $X(u)^2, Y(v)^2 \preceq I$. Let
$P = \mathrm{diag}(p_\alpha)$. Expanding
$\langle\Phi_V|(u\cdot\sigma)\otimes I|\Phi_V\rangle$ and
$\langle\Phi_V|(u\cdot\sigma)\otimes(v\cdot\sigma)|\Phi_V\rangle$ in
the Schmidt basis, and using that the Werner marginals vanish and the
correlation matrix is $tI_3$:
\begin{equation}
\mathrm{Tr}\left(P\,X(u)\right) = \mathrm{Tr}\left(P\,Y(v)\right) = 0,
\qquad
\sum_{\alpha,\beta}\sqrt{p_\alpha p_\beta}\;
X(u)_{\alpha\beta}\,Y(v)_{\alpha\beta}
= t\,(u\cdot v).
\tag{A.1}\label{eq:corr}
\end{equation}
The second identity follows by direct expansion:
\begin{multline*}
\langle\Phi_V|(u\cdot\sigma\otimes I)\otimes(v\cdot\sigma\otimes I)
|\Phi_V\rangle
= \sum_{\alpha,\beta}\sqrt{p_\alpha p_\beta}\,
\langle x_\beta|(u\cdot\sigma\otimes I)|x_\alpha\rangle\,
\langle y_\beta|(v\cdot\sigma\otimes I)|y_\alpha\rangle\\
= \sum_{\alpha,\beta}\sqrt{p_\alpha p_\beta}\,
X(u)_{\beta\alpha}\,Y(v)_{\beta\alpha},
\end{multline*}
and relabeling $(\alpha,\beta)$ gives the symmetric form displayed in
\eqref{eq:corr}; note that \emph{both} matrix elements appear
unconjugated, which is why the transpose (rather than the adjoint) of
one factor appears in the trace form used in Appendix~A.5. No reality
or rank restriction on the Schmidt vectors is used.

\subsection{The row-Bessel remainder estimate}

Let $R(u,v)$ be the part of the left side of \eqref{eq:corr} in which
the index pair $(\alpha,\beta)$ is \emph{not} contained in
$\{1,2\}^2$. We prove
\begin{equation}
|R(u,v)| \;\le\; \varepsilon + 2sz
\qquad\text{for all unit } u, v,
\tag{A.2}\label{eq:remainder}
\end{equation}
which is the operator-norm bound of Lemma~\ref{lem:remainder} for the
matrix $R_{jk} = R(e_j, e_k)$ (the bilinear form determines the
matrix, and the bound is uniform over unit vectors).

Fix unit $u,v$ and abbreviate $X = X(u)$, $Y = Y(v)$. Split $R$ into
the \emph{cross} part (one index in $\{1,2\}$, the other $\ge 3$) and
the \emph{tail} part (both indices $\ge 3$).

\emph{Cross part.} Hermiticity of $X$ and $Y$ pairs the top--tail and
tail--top terms:
\[
R_{\mathrm{cross}} = 2\,\mathrm{Re}\!\!\sum_{i\le2,\ \alpha\ge3}
\sqrt{p_i p_\alpha}\,X_{i\alpha}Y_{i\alpha} .
\]
By the Cauchy--Schwarz inequality,
$|R_{\mathrm{cross}}| \le 2\sqrt{B_X B_Y}$ with
$B_X = \sum_{i\le2,\alpha\ge3} p_i p_\alpha |X_{i\alpha}|^2$. The key
step is a \emph{shared row budget}: ordering gives $p_\alpha \le z$
for every $\alpha \ge 3$, and Bessel's inequality for the rows of the
contraction $X$ gives
\[
\sum_{\alpha\ge3}|X_{i\alpha}|^2
\;\le\; \sum_{\beta}|X_{i\beta}|^2
\;=\; (X^2)_{ii} \;\le\; 1 .
\]
Hence $B_X \le z(p_1+p_2) = zs$, and identically $B_Y \le zs$, so
$|R_{\mathrm{cross}}| \le 2zs$. Note that the bound scales with the
\emph{largest tail probability} $z$, not with $\sqrt{\varepsilon}$:
this is what makes the final inequality \eqref{eq:G} strong enough for
the certificate.

\emph{Tail part.} If $\varepsilon > 0$, let
$|\Phi_{\mathrm{tail}}\rangle = \varepsilon^{-1/2}
\sum_{\alpha\ge3}\sqrt{p_\alpha}\,|x_\alpha\rangle|y_\alpha\rangle$
be the normalized tail of the Schmidt sum. Then the tail--tail part of
the sum equals exactly
$\varepsilon\,\langle\Phi_{\mathrm{tail}}|
(u\cdot\sigma\otimes I)\otimes(v\cdot\sigma\otimes I)
|\Phi_{\mathrm{tail}}\rangle$, an expectation of an observable of
operator norm one; its absolute value is therefore at most
$\varepsilon$. If $\varepsilon = 0$ the tail part vanishes. Adding the
two parts proves \eqref{eq:remainder}.

\subsection{The orientation determinant cap}

\noindent\textbf{Lemma 1.} \emph{If $r \mapsto c + Lr$ is the affine
Bloch form of a qubit channel, then $|\det L| \le 1$; and
$\det L < 0$ implies $|\det L| \le 1/27$.}

\medskip

\emph{Proof.} Every qubit channel contracts the trace distance, and on
Bloch vectors the trace distance is the Euclidean distance, so
$\|L\|_{\mathrm{op}} \le 1$ and $|\det L| \le 1$.

Suppose $\det L < 0$; then $L$ is invertible. Polar-decompose
$L = OP$ with $P > 0$ symmetric and $O$ orthogonal; $\det L < 0$
makes $O$ improper ($\det O = -1$). For every proper rotation
$R \in SO(3)$, consider the new channel obtained by preceding the
given channel with the unitary rotation whose Bloch action is $R^T$
and following it with the unitary rotation whose Bloch action is
$ORO^T$. The conjugated rotation $ORO^T$ is proper
($\det = (\det O)^2\det R = 1$), so both pre- and post-processing are
legitimate unitary channels. The composite has affine part
$r \mapsto (ORO^T)c + (ORO^T)(OP)R^T r$. Averaging over the Haar
measure of $SO(3)$ (a convex mixture of channels, hence a channel):
the translation averages to
$\int (ORO^T)c\,dR = O\int R\,(O^Tc)\,dR = 0$ because
$\int R v\,dR = 0$ for every fixed vector $v$, and the linear part
becomes
\[
\int (ORO^T)(OP)R^T\,dR
\;=\; O\int R P R^T\,dR
\;=\; \alpha\,O,
\qquad
\alpha = \frac{\mathrm{Tr}\,P}{3},
\]
using the standard twirl identity
$\int RPR^T dR = (\mathrm{Tr}P/3)I$. A final proper rotation on each
side brings $O$ to $\mathrm{diag}(1,1,-1)$, so we obtain a unital
qubit channel with diagonal Bloch matrix
$\mathrm{diag}(\alpha,\alpha,-\alpha)$. A unital qubit channel with
diagonal Bloch multipliers $(\lambda_1,\lambda_2,\lambda_3)$ is a
Pauli channel with probabilities
$q_0 = (1+\lambda_1+\lambda_2+\lambda_3)/4$ and cyclically
$q_j = (1+\lambda_j-\lambda_k-\lambda_l)/4$; equivalently, its
normalized Choi state $(\mathrm{id}\otimes\Lambda)(\Phi^+)$ is
diagonal in the Bell basis with eigenvalues exactly the $q_\mu$.
Complete positivity is exactly $q_\mu \ge 0$ for all four. For
$(\alpha,\alpha,-\alpha)$ the four values $q_\mu$ are
\[
\frac{1+\alpha}{4},\quad \frac{1+\alpha}{4},\quad
\frac{1+\alpha}{4},\quad \frac{1-3\alpha}{4},
\]
so complete positivity forces $1 - 3\alpha \ge 0$, i.e.\
$\alpha \le 1/3$. Finally, by the arithmetic--geometric mean
inequality applied to the three positive eigenvalues of $P$,
\[
|\det L| = \det P \le
\left(\frac{\mathrm{Tr}\,P}{3}\right)^{3} = \alpha^3 \le \frac{1}{27}.
\]
\hfill$\blacksquare$

\subsection{The channel construction and the orientation polynomial}

\emph{The seed channels.} Restrict $U_A$ to the span of the first two
Schmidt indices; this gives an isometry
$W_A : \mathbb{C}^2 \to A\otimes A'$, and
\[
\Lambda_A(\omega) \;=\; \mathrm{Tr}_{A'}\!\left(W_A\,\omega\,
W_A^\dagger\right)
\]
is a qubit channel from the two-dimensional \emph{seed} space (spanned
by Schmidt labels $1,2$) to the physical qubit $A$. Let its Bloch form
be $r \mapsto a_0 + Ar$. Its adjoint compresses observables: for the
$2\times2$ restriction $\bar X(u)$ of $X(u)$ to the seed space,
\begin{equation}
\bar X(u) \;=\; \Lambda_A^\dagger(u\cdot\sigma)
\;=\; (u\cdot a_0)\,I_2 + (A^T u)\cdot\sigma .
\tag{A.3}\label{eq:adjoint}
\end{equation}
Identically on the $B$ side with Bloch form $r \mapsto b_0 + Br$ and
$\bar Y(v) = (v\cdot b_0)I_2 + (B^T v)\cdot\sigma$.

\emph{The seed state.} The normalized leading-two weight matrix is
$\tau = \mathrm{diag}(p_1, p_2)/s$, a qubit state with Bloch vector
$\delta e_3$ in the seed basis, $\delta = (p_1-p_2)/s$, and
$\tau^{1/2} = \mathrm{diag}(u_+, u_-)$ with
$u_\pm = \sqrt{(1\pm\delta)/2}$.

\emph{The top block.} The part of \eqref{eq:corr} with both indices in
$\{1,2\}$ is
\[
\sum_{i,j\le2}\sqrt{p_ip_j}\,\bar X_{ij}\bar Y_{ij}
= s\,\mathrm{Tr}\!\left(\tau^{1/2}\bar X\,\tau^{1/2}\bar Y^{T}\right).
\]
Write $\bar X = x_0 I + x\cdot\sigma$ and
$\bar Y = y_0 I + y\cdot\sigma$ per \eqref{eq:adjoint}, so
$\bar Y^T = y_0 I + y_1\sigma_1 - y_2\sigma_2 + y_3\sigma_3$
(transposition in the seed basis fixes $\sigma_1,\sigma_3$ and negates
$\sigma_2$). A direct $2\times2$ computation with
$D_\tau = \mathrm{diag}(u_+,u_-)$ gives
\[
\mathrm{Tr}\!\left(D_\tau \bar X D_\tau \bar Y^{T}\right)
= x_0y_0 + \delta(x_0y_3 + x_3y_0) + x_3y_3
+ \sqrt{1-\delta^2}\,(x_1y_1 - x_2y_2).
\]
Completing the product,
\[
= (x_0+\delta x_3)(y_0+\delta y_3)
+ \sqrt{1-\delta^2}\,x_1y_1
- \sqrt{1-\delta^2}\,x_2y_2
+ (1-\delta^2)\,x_3y_3 .
\]
Substituting $x_0 = u\cdot a_0$, $x = A^Tu$, and likewise for $y$:
$x_0 + \delta x_3 = u\cdot(a_0 + \delta Ae_3) = u\cdot a$, where
$a := a_0 + A(\delta e_3)$ is exactly the output Bloch vector
$\Lambda_A(\tau)$ of the seed state; similarly
$b := b_0 + B(\delta e_3)$. The remaining three terms assemble into
$u^T A\,C_\delta\,B^T v$ with
\[
C_\delta = \mathrm{diag}\!\left(\sqrt{1-\delta^2},\,
-\sqrt{1-\delta^2},\, 1-\delta^2\right),
\qquad
\det C_\delta = -(1-\delta^2)^2 .
\]
Hence, as bilinear forms in $(u,v)$, the full identity
\eqref{eq:corr} becomes
\begin{equation}
t\,I_3 \;=\; s\,A C_\delta B^T \;+\; s\,ab^T \;+\; R ,
\tag{A.4}\label{eq:decomp}
\end{equation}
with $R$ bounded by \eqref{eq:remainder}.

\emph{The marginal cancellation.} The full one-body marginals vanish:
$0 = \mathrm{Tr}(P X(u)) = s\,(u\cdot a) + (\text{tail
contribution})$, where the tail contribution is $\varepsilon$ times
the expectation of $u\cdot\sigma$ in the normalized tail state, of
absolute value at most $1$. Hence $|u\cdot a| \le \varepsilon/s$ for
every unit $u$, i.e.\ $\|a\| \le \varepsilon/s$; identically
$\|b\| \le \varepsilon/s$, so
$\|s\,ab^T\|_{\mathrm{op}} \le \varepsilon^2/s$. With
\eqref{eq:remainder} this proves Lemma~\ref{lem:remainder}:
\[
\|D - tI_3\|_{\mathrm{op}} \le
\eta = \varepsilon + 2sz + \frac{\varepsilon^2}{s},
\qquad D = sAC_\delta B^T .
\]

\emph{The determinant sandwich.} Suppose $t > \eta$. Every matrix on
the segment from $tI_3$ to $D$ is within $\eta < t$ of $tI_3$ in
operator norm, hence invertible; determinants are continuous and
nonzero along the segment, so $\det D > 0$, and moreover every
singular value of $D$ is at least $t - \eta$, giving
\begin{equation}
\det D \;\ge\; (t-\eta)^3 \;>\; 0 .
\tag{A.5}\label{eq:detlow}
\end{equation}
On the other hand $\det D = s^3\,\det A\,\det C_\delta\,\det B =
-s^3(1-\delta^2)^2\det A\det B$. Positivity of $\det D$ forces
$\det A\det B < 0$: exactly one of the two channels is
orientation-reversing. Apply Lemma~\ref{lem:orientation}'s cap $1/27$
to that channel and the trivial cap $1$ to the other:
\begin{equation}
\det D \;\le\; \frac{s^3(1-\delta^2)^2}{27} .
\tag{A.6}\label{eq:dethigh}
\end{equation}

\emph{Clearing denominators.} $1-\delta^2 = 4p_1p_2/s^2$, so
\eqref{eq:detlow}--\eqref{eq:dethigh} give
$(t-\eta)^3 \le 16p_1^2p_2^2/(27 s)$; multiplying by $27 s^3 > 0$ and
writing $N = s(t-\eta) = s(t-\varepsilon-2sz)-\varepsilon^2$:
\[
G = 16\,p_1^2p_2^2s^2 - 27N^3 \;\ge\; 0 .
\]
If instead $t \le \eta$, then $N \le 0$, so $-27N^3 \ge 0$ and
$G \ge 0$ automatically. Inequality \eqref{eq:G} therefore holds for
every physical spectrum, with both cases covered explicitly.
\hfill$\blacksquare$

\subsection{The symmetric-extension spectral lemma and the Mirsky
transfer}

\noindent\textbf{Lemma 4.} \emph{Let $\sigma_{XY}$ be a state,
$\dim Y = 2$, $\sigma_Y = I/2$, admitting a symmetric extension in
$Y$. With ordered eigenvalues $q_1 \ge q_2 \ge q_3 \ge \cdots$ and
$Q_\sigma = \sum_{i\ge4}q_i$:
$\sqrt{q_1} \le \sqrt{q_2}+\sqrt{q_3}+\sqrt{Q_\sigma}$.}

\medskip

\emph{Proof.} If $q_1 \le 1/2$: the square of the right side is at
least $q_2 + q_3 + Q_\sigma = 1 - q_1 \ge q_1$, done. Assume
$q_1 > 1/2$ and proceed in eight steps.

\textbf{(1) Compression preserves the extension.} Let $|e\rangle$ be a
leading eigenvector of $\sigma_{XY}$ and $X_0 \subseteq X$ the Schmidt
support of $|e\rangle$ on the $X$ side; since $Y$ is a qubit,
$\dim X_0 \le 2$. Let $\Pi = P_{X_0}\otimes I_Y$ and
$\omega = \Pi\sigma_{XY}\Pi$ (unnormalized). If $\sigma_{XYY'}$ is a
symmetric extension of $\sigma_{XY}$, then
$(P_{X_0}\otimes I_Y\otimes I_{Y'})\sigma_{XYY'}
(P_{X_0}\otimes I_Y\otimes I_{Y'})$ is positive, symmetric under the
$Y\leftrightarrow Y'$ swap (the projector acts only on $X$), and
reduces to $\omega$ on $XY$; so $\omega$ has a symmetric extension.

\textbf{(2) The leading eigenvalue survives.}
$|e\rangle \in X_0\otimes Y$ by construction, so
$\Pi|e\rangle = |e\rangle$ and
$\omega|e\rangle = \Pi\sigma\Pi|e\rangle = \Pi\sigma|e\rangle =
q_1|e\rangle$: the compressed operator keeps $|e\rangle$ as an
eigenvector with eigenvalue exactly $q_1$.

\textbf{(3) Interlacing caps the rest.} $\omega$ acts on the (at most)
four-dimensional space $X_0\otimes Y$; write its ordered eigenvalues
as $q_1, r_2, r_3, r_4$. The Cauchy interlacing inequalities for the
compression $\Pi\sigma\Pi$ give $0 \le r_j \le q_j$ for $j = 2,3,4$.

\textbf{(4) The compressed marginal obeys
$0 \preceq \omega_Y \preceq I/2$.}
$\omega_Y = \mathrm{Tr}_{X_0}(\Pi\sigma\Pi) \preceq
\mathrm{Tr}_X\sigma = \sigma_Y = I/2$, since removing the projector
only adds a positive operator. Equality is \emph{not} claimed:
compression generally shrinks the marginal, and the proof below uses
only the inequality.

\textbf{(5) The two-qubit criterion applies homogeneously.} The
theorem of Chen--Ji--Kribs--L\"utkenhaus--Zeng \cite{CJKLZ} states
that a two-qubit \emph{state} $\omega$ has a symmetric extension in
$Y$ if and only if
$\mathrm{Tr}\,\omega_Y^2 \ge \mathrm{Tr}\,\omega^2 -
4\sqrt{\det\omega}$. Both sides are homogeneous of degree two under
$\omega \mapsto \lambda\omega$, so the criterion holds verbatim for
the unnormalized positive operator $\omega$ (apply it to
$\omega/\mathrm{Tr}\,\omega$ and multiply through). By step (1)
$\omega$ has a symmetric extension, so the inequality holds.

\textbf{(6) Eliminating the marginal.} Let $\lambda, \mu$ be the
eigenvalues of $\omega_Y$; step (4) gives $\lambda, \mu \le 1/2$, and
$\lambda + \mu = \mathrm{Tr}\,\omega =: w$. Then
\[
\lambda^2 + \mu^2 \;\le\; \frac14 + \left(w - \frac12\right)^{2},
\]
because the difference of the two sides is
$2(\tfrac12-\lambda)(\tfrac12-\mu) \ge 0$. The eigenvalues of $\omega$
are $q_1, r_2, r_3, r_4$, so
$\mathrm{Tr}\,\omega^2 = q_1^2+r_2^2+r_3^2+r_4^2$ and
$\det\omega = q_1r_2r_3r_4$. Combining with step (5),
\[
\frac14 + \left(w-\frac12\right)^{2}
- \mathrm{Tr}\,\omega^2 + 4\sqrt{\det\omega}
\;\ge\;
\mathrm{Tr}\,\omega_Y^2 - \mathrm{Tr}\,\omega^2 + 4\sqrt{\det\omega}
\;\ge\; 0 .
\]
Doubling the left side and using the algebraic identity
$2(w-\tfrac12)^2 + \tfrac12 = w^2 + (1-w)^2$, this says exactly
$\Phi_{q_1}(r_2,r_3,r_4) \ge 0$, where, for $A = q_1$ and
$w = A+x+y+z$,
\[
\Phi_A(x,y,z) = (A{+}x{+}y{+}z)^2 - 2(A^2{+}x^2{+}y^2{+}z^2)
+ 8\sqrt{Axyz} + (1{-}A{-}x{-}y{-}z)^2 .
\]

\textbf{(7) Monotonicity replaces $r_j$ by $q_j$.} For $A > 1/2$ the
function $\Phi_A$ is coordinatewise nondecreasing on
$\{x,y,z \ge 0,\ A+x+y+z \le 1\}$: away from a zero coordinate,
\[
\partial_x \Phi_A
= 2(A{+}x{+}y{+}z) - 4x + 4\sqrt{Ayz/x} - 2(1{-}A{-}x{-}y{-}z)
= 4(A+y+z) - 2 + 4\sqrt{Ayz/x} ,
\]
which is positive because $A > 1/2$ gives $4A - 2 > 0$; the boundary
case of a vanishing coordinate follows by continuity, and the same
computation applies in $y$ and $z$. Step (3) gives $r_j \le q_j$ for
$j=2,3$ and $r_4 \le q_4 \le Q_\sigma$, so
\[
\Phi_{q_1}(q_2, q_3, Q_\sigma) \;\ge\;
\Phi_{q_1}(r_2, r_3, r_4) \;\ge\; 0 .
\]

\textbf{(8) Factorization.} The four arguments
$q_1, q_2, q_3, Q_\sigma$ sum to one, so the last square in $\Phi$
vanishes. Substituting $a = \sqrt{q_1}$, $b = \sqrt{q_2}$,
$c = \sqrt{q_3}$, $d = \sqrt{Q_\sigma}$, the remaining expression
factors exactly:
\[
\Phi = (a{+}b{+}c{-}d)(a{+}b{-}c{+}d)(a{-}b{+}c{+}d)
({-}a{+}b{+}c{+}d) .
\]
Since $q_1 > 1/2$ exceeds each of $q_2, q_3, Q_\sigma$ (they sum to
$1-q_1 < 1/2$), $a$ exceeds each of $b, c, d$, so the first three
factors are strictly positive. Nonnegativity of the product forces
$-a+b+c+d \ge 0$, which is the lemma. \hfill$\blacksquare$

\medskip

\noindent\textbf{Application and Mirsky transfer
(Proposition~\ref{prop:spectral}).} Let $M$ be the $8\times8$
coefficient matrix of $|\Phi_V\rangle$ across $AA'\,|\,BB'$ (so the
$p_\alpha$ are the squared singular values of $M$), and split
$M = S + E$ along the physical singlet/triplet decomposition of the
$AB$ qubit pair. The singlet weight of $W(f_0)$ is $f_0$, so
$\|E\|_F^2 = 1/200$ exactly.

The normalized state $P_{\mathrm{sym}}|\Phi_V\rangle/\|\cdot\|$
(projection onto the physical triplet subspace; its coefficient
matrix is $S/\|S\|_F$) satisfies
$\mathrm{Swap}_{AB}\,P_{\mathrm{sym}}|\Phi_V\rangle =
P_{\mathrm{sym}}|\Phi_V\rangle$, because the swap acts as the identity
on the symmetric subspace. Its physical reduction is the normalized
triplet part of $W(f_0)$, namely $P_{\mathrm{sym}}/3$, whose qubit
marginals are $I/2$. Let $\sigma_{A'A}$ be its reduction to $A'A$.
Tracing out $B'$ only, swap symmetry makes the resulting state on
$A'AB$ symmetric under exchanging the qubits $A$ and $B$; that state
is therefore a symmetric extension of $\sigma_{A'A}$ in its qubit
system $A$, and $\sigma_A = I/2$. Lemma~\ref{lem:spectral} applies
with $X = A'$ (dimension $4$), $Y = A$, and eigenvalues
$q_i = t_i^2/\|S\|_F^2$, where $t_1 \ge t_2 \ge \cdots$ are the
singular values of $S$ (the eigenvalues of the $A'A$ reduction of a
pure state are the squared Schmidt coefficients, i.e.\ the squared
singular values of the coefficient matrix). Multiplying the lemma
through by $\|S\|_F$:
\begin{equation}
t_1 - t_2 - t_3 \;\le\; \Bigl(\sum_{i\ge4}t_i^2\Bigr)^{1/2} .
\tag{A.7}\label{eq:lemS}
\end{equation}
Let $s_i = \sqrt{p_i}$ be the singular values of $M$. Mirsky's
inequality \cite{Mirsky} --- for any two matrices, the Euclidean
distance between their ordered singular-value vectors is at most the
Frobenius distance between the matrices; the statement we use is
reproduced, with attribution to Mirsky, as Theorem~2.3 of the freely
available survey of Markus \cite{Markus} --- gives
\begin{equation}
\sum_i (s_i - t_i)^2 \;\le\; \|M - S\|_F^2 = \frac{1}{200} .
\tag{A.8}\label{eq:mirsky}
\end{equation}
Set $x^2 = \sum_{i\le3}(s_i-t_i)^2$ and
$w^2 = \sum_{i\ge4}(s_i-t_i)^2$. Then, using \eqref{eq:lemS}, the
triangle inequality in $\ell^2$ for the tail
($\sqrt{\sum_{i\ge4}t_i^2} \le \sqrt{\sum_{i\ge4}s_i^2} + w
= \sqrt{Q} + w$), and Cauchy--Schwarz
($\sqrt3\,x + w \le 2\sqrt{x^2+w^2}$):
\begin{align*}
s_1 &\le t_1 + |s_1{-}t_1|
\le t_2 + t_3 + \sqrt{Q} + w + |s_1{-}t_1|\\
&\le s_2 + s_3 + \sqrt{Q} + w + \sqrt3\,x
\le s_2 + s_3 + \sqrt{Q} + 2\sqrt{x^2+w^2}.
\end{align*}
With \eqref{eq:mirsky},
$2\sqrt{x^2+w^2} \le 2\sqrt{1/200} = 1/\sqrt{50}$, which is
inequality \eqref{eq:spectral}. \hfill$\blacksquare$

\subsection{Tail cap, entropy merge, and the $p_1 \le 1/2$ case}

$Q = \sum_{\alpha=4}^{r} p_\alpha$ has at most five terms
($r \le 8$), each at most $p_3$ by ordering; hence
$0 \le Q \le 5p_3$, inequality \eqref{eq:tailcap}. For
\eqref{eq:entropymerge}: $h(x) = -x\log_2x$ is concave with
$h(0) = 0$, hence subadditive: $h(x+y) \le h(x) + h(y)$ for
$x, y \ge 0$, $x+y \le 1$; iterating over the tail terms,
$\sum_{\alpha\ge4}h(p_\alpha) \ge h(Q)$, so
$H(p) \ge h(p_1)+h(p_2)+h(p_3)+h(Q)$. Finally, if $p_1 \le 1/2$, then
the partial sums of $p$ are dominated by those of
$(\tfrac12,\tfrac12,0,\ldots,0)$ --- $p_1 \le \tfrac12$ and
$p_1 + p_2 \le 1$ --- so $p$ is majorized by it, and Schur concavity
of Shannon entropy gives $H(p) \ge H(\tfrac12,\tfrac12) = 1$.

\subsection{The endpoint isometry and its spectrum}

Feasibility: the map $V|i\rangle = |\Psi_i\rangle_{A'B'}$ sends an
orthonormal basis to four orthonormal vectors, so $V^\dagger V = I_4$
and $|\psi\rangle = (I_{AB}\otimes V)\sum_i\sqrt{\lambda_i}
|\Psi_i\rangle_{AB}|i\rangle_{A'}$ is a purification of $W(1/100)$.

Spectrum: the $AA'$ marginal $\rho_{AA'}$ of $|\psi\rangle$ is an
$8\times8$ positive matrix of rank at most four with entries in
$\mathbb{Q}[\sqrt{33}]$ after inserting $\sqrt{\lambda_0} = 1/10$ and
$\sqrt{\lambda_{1,2,3}} = \sqrt{33}/10$. The two-qubit ancillary Bell
vectors are embedded in the first $\mathbb{C}^2\otimes\mathbb{C}^2$
block of the $\mathbb{C}^4\otimes\mathbb{C}^4$ ancilla space of
Section~\ref{sec:reduction}, so in that basis
$\rho_{AA'} = \rho_{4\times4} \oplus 0_4$: the four padding dimensions
contribute four exact zero eigenvalues. This is the embedding the
verification program uses when it forms $\rho_{AA'}$. The claim is
that its nonzero spectrum is $\{s, s, s, 1-3s\}$ with
$s = (17-\sqrt{33})/200$. A multiset of at most four nonzero
eigenvalues is determined by its first four power sums (Newton's
identities), and the verification program checks, in exact arithmetic
over $\mathbb{Q}[\sqrt{33}]$,
\[
\mathrm{Tr}\,\rho_{AA'}^{\,k} \;=\; 3s^k + (1-3s)^k,
\qquad k = 1, 2, 3, 4,
\]
together with $3s + (1-3s) = 1$. This spectrum identification is
counted among the computer-assisted steps. The resulting feasible
value is $U_w = -3s\log_2 s - (1-3s)\log_2(1-3s)$, and the directed
rational certificate (Appendix~B) proves $U_w \le U_+$. The
certificate encloses $s$ in a dyadic bracket
$[s_{\mathrm{lo}}, s_{\mathrm{hi}}]$ and evaluates the entropy
expression at $s_{\mathrm{hi}}$; this yields an upper bound on $U_w$
because the map $s \mapsto -3s\log_2 s - (1-3s)\log_2(1-3s)$ is
increasing on $0 < s < 1/4$ --- its derivative is
$3\log_2\!\bigl((1-3s)/s\bigr) > 0$ there since $1-3s > 1/4 > s$ ---
and the certified bracket for
$s = (17-\sqrt{33})/200 \approx 0.0563$ lies inside $(0, 1/4)$.

\subsection{The symmetric-support endpoint}

\noindent\textbf{Claim.} $E_P(W(0)^{\otimes n}) = n$ for every
$n \ge 1$; hence $E_P^\infty(W(0)) = 1$.

\medskip

\emph{Proof.} $W(0) = P_{\mathrm{sym}}/3$ is supported on the
symmetric subspace and has marginals $S(A) = S(B) = 1$ (its qubit
marginal is $I/2$).

\emph{Lower bound.} A state supported on the symmetric subspace has a
purification invariant under the physical swap $A \leftrightarrow B$
(purify $P_{\mathrm{sym}}/3$ by its canonical purification
$\sum_i\sqrt{1/3}\,|v_i\rangle_{AB}|i\rangle_E$ with
$\{|v_i\rangle\}$ an orthonormal basis of the symmetric subspace; the
swap fixes each $|v_i\rangle$, hence the whole vector). Every
extension $\omega^{ABE'}$ of $W(0)$ arises from that purification by
a channel acting on the purifying system alone, which commutes with
the swap on $AB$; so every extension satisfies
$\mathrm{Swap}_{AB}\,\omega\,\mathrm{Swap}_{AB} = \omega$, and in
particular $S(E'|A)_\omega = S(E'|B)_\omega$, where
$S(E'|A) = S(AE') - S(A)$ is the conditional entropy. Weak
monotonicity of the von Neumann entropy (equivalent to strong
subadditivity) states $S(E'|A) + S(E'|B) \ge 0$; by the symmetry just
derived, $S(E'|A) \ge 0$, i.e.\ $S(AE') \ge S(A) = 1$. Now let
$|\Phi\rangle$ on $AA'BB'$ be any purification of $W(0)$ and take
$E' = A'$: the reduction to $ABA'$ is an extension, so
$S(AA')_\Phi \ge 1$, giving $E_P(W(0)) \ge 1$.

\emph{Upper bound.} The purification with trivial $B'$ has
$S(AA') = S(BB') = S(B) = 1$ (purity across the cut), so
$E_P(W(0)) \le 1$.

\emph{Tensor powers.} $W(0)^{\otimes n}$ is supported on the $n$-fold
tensor product of symmetric subspaces, which is contained in the
fixed-point space of the composite swap
$\mathrm{Swap}_{A^nB^n} = \bigotimes_{i}\mathrm{Swap}_{A_iB_i}$
exchanging the entire $A$-string with the entire $B$-string. The
argument above applies verbatim with $A \to A^n$, $B \to B^n$ and
$S(A^n) = n$, giving $E_P(W(0)^{\otimes n}) = n$. The general
statement, for any state supported on a symmetric or antisymmetric
subspace with $|A| = |B|$, is Proposition~7 of Christandl--Winter
\cite{CW}. \hfill$\blacksquare$

\renewcommand{\thesection}{A\texorpdfstring{$'$}{'}}
\setcounter{subsection}{0}
\section{The finite-ensemble inequality for
\texorpdfstring{$E_P^\infty$}{EP-infinity}}

\noindent\textbf{Lemma A$'$.} \emph{For every finite ensemble
$\rho = \sum_{i=1}^{k} p_i\rho_i$ of bipartite states on $AB$,}
\[
E_P^\infty(\rho) \;\le\; \sum_i p_i\,E_P^\infty(\rho_i) + \chi,
\qquad
\chi = S(\rho) - \sum_i p_i S(\rho_i) \;\ge 0.
\]
\emph{Equivalently, $E_P^\infty - S$ is convex on ensembles.}

\medskip

This inequality is due to Chen and Winter (\cite{ChenWinter},
Theorem~3 and Corollary~4); their preprint presents the proof as a
summary outline resting on published covering and operational
results. Since \cite{ChenWinter} is unpublished and this inequality
is load-bearing here, we give a complete proof from the following
published ingredients.

\medskip

\noindent\textbf{(I1) Operational identity} (\cite{THLD}).
$E_P^\infty(\rho) = E_{LOq}(\rho)$, where $E_{LOq}(\rho)$ is the
least $R$ such that for every $\epsilon > 0$ and all large $n$ there
is a protocol using $\lceil nR \rceil$ maximally entangled qubit
pairs (``ebits''), local operations, and $o(n)$ \emph{qubits of
quantum communication} --- that is the class LOq as defined in
\cite{THLD} --- whose output is within trace distance $\epsilon$ of
$\rho^{\otimes n}$. Sending classical bits is a special case of
sending qubits, so a protocol that uses only $o(n)$ classical bits,
as ours does ($O(\sqrt n)$ bits), is admissible in LOq. (The converse
direction of this identity --- no protocol can beat $E_P^\infty$ ---
uses the asymptotic continuity of $E_P$, proved in \cite{THLD}; we
use the identity as a black box in both directions.)

\medskip

\noindent\textbf{(I2) Entanglement dilution at entropy rate with
sublinear communication} (\cite{LoPopescu}). For any bipartite pure
state $|\psi\rangle$, the state $|\psi\rangle^{\otimes m}$ can be
prepared from $m(S(\psi_A) + \delta)$ ebits by local operations and
$O(\sqrt m)$ classical bits, with trace-distance error vanishing as
$m \to \infty$. (This achievability statement is \cite{LoPopescu}.
The companion results \cite{HarrowLo, HaydenWinter} are converses:
\cite{HarrowLo} proves a matching $\Omega(\sqrt m)$ \emph{lower}
bound on the classical communication of dilution, and
\cite{HaydenWinter} treats communication costs of entanglement
transformations in general. They show (I2) is tight but are not used
in the construction.)

\medskip

\noindent\textbf{(I3) Typicality estimates} (standard; in the forms
used by \cite{WinterCMP, WinterIT}). For a state $\tau$ with spectrum
$\{\lambda_j\}$ and $\delta' > 0$, the typical projector
$\Pi^n_{\tau,\delta'}$ (onto eigenvectors of $\tau^{\otimes n}$ with
eigenvalue in $[2^{-n(S(\tau)+\delta')}, 2^{-n(S(\tau)-\delta')}]$)
satisfies, for $n$ large:
$\mathrm{Tr}(\tau^{\otimes n}\Pi) \ge 1-\epsilon$;
$\mathrm{rank}\,\Pi \le 2^{n(S(\tau)+\delta')}$; and
$\Pi\,\tau^{\otimes n}\Pi \le 2^{-n(S(\tau)-\delta')}\Pi$. For a
string $i^n = i_1\ldots i_n$ that is typical for the distribution $p$
(letter frequencies within $\delta'$ of $p$), the conditional state
$\rho_{i^n} = \rho_{i_1}\otimes\cdots\otimes\rho_{i_n}$ has a
conditional typical projector $\Pi_{i^n}$ with the analogous three
properties around $\bar S = \sum_i p_iS(\rho_i)$ (up to a
$\delta'$-linear correction).

\medskip

\noindent\textbf{(I4) Gentle operator lemma} (\cite{WinterIT};
subnormalized form as in \cite{WinterCMP}). If $\tau \ge 0$ with
$\mathrm{Tr}\,\tau \le 1$ and $0 \le \Lambda \le I$ satisfy
$\mathrm{Tr}(\tau\Lambda) \ge \mathrm{Tr}\,\tau - \kappa$, then
$\|\tau - \sqrt\Lambda\,\tau\sqrt\Lambda\|_1 \le 2\sqrt\kappa$. For a
normalized state with $\mathrm{Tr}(\tau\Lambda) \ge 1-\epsilon$ this
is the usual statement with bound $2\sqrt\epsilon$. (\cite{WinterIT}
proves this with the constant $\sqrt{8\kappa}$, extended verbatim to
subnormalized $\tau$ in \cite{WinterCMP}, Lemma~4; that constant
would also suffice everywhere below. The sharper constant
$2\sqrt\kappa$ used here has a four-line proof, which we include to
keep the constant self-contained. Write
$\tau - \sqrt\Lambda\,\tau\sqrt\Lambda
= (I-\sqrt\Lambda)\,\tau + \sqrt\Lambda\,\tau\,(I-\sqrt\Lambda)$. By
the Cauchy--Schwarz inequality $\|AB\|_1 \le \|A\|_2\|B\|_2$ for the
trace norm, splitting $\tau = \sqrt\tau\cdot\sqrt\tau$,
$\|(I-\sqrt\Lambda)\tau\|_1 \le
\bigl(\mathrm{Tr}[(I-\sqrt\Lambda)^2\tau]\bigr)^{1/2}
(\mathrm{Tr}\,\tau)^{1/2}$ and
$\|\sqrt\Lambda\,\tau(I-\sqrt\Lambda)\|_1 \le
\bigl(\mathrm{Tr}[\Lambda\tau]\bigr)^{1/2}
\bigl(\mathrm{Tr}[(I-\sqrt\Lambda)^2\tau]\bigr)^{1/2}$. Since
$0 \le \Lambda \le I$ gives $\Lambda \le \sqrt\Lambda$, we have
$(I-\sqrt\Lambda)^2 = I - 2\sqrt\Lambda + \Lambda \le I - \Lambda$,
so $\mathrm{Tr}[(I-\sqrt\Lambda)^2\tau] \le
\mathrm{Tr}\,\tau - \mathrm{Tr}(\tau\Lambda) \le \kappa$, while
$\mathrm{Tr}\,\tau \le 1$ and $\mathrm{Tr}[\Lambda\tau] \le 1$; each
term is at most $\sqrt\kappa$.)

\medskip

\noindent\textbf{(I5) Operator sampling (operator Chernoff bound)}
(\cite{AW}; applied in covering form in \cite{WinterCMP}, Theorem~2,
and \cite{GPW}, Proposition~II.2). Let $X_1,\ldots,X_K$ be i.i.d.\
random operators on a $d$-dimensional space with $0 \le X_j \le I$
and mean $M = \mathbb E X_j \ge a\,\Pi'$ on the support projector
$\Pi'$ of $M$, $a > 0$. Then for $0<\epsilon<1$,
\[
\Pr\!\left[\ \frac1K\sum_j X_j \not\in
[(1-\epsilon)M,\ (1+\epsilon)M]\ \right]
\;\le\; 2d\,\exp\!\left(-\frac{K\epsilon^2 a}{2\ln 2}\right).
\]

\medskip

\noindent\textbf{Proof of Lemma A$'$.} Fix
$\delta, \delta', \epsilon > 0$. Write
$\rho^{\otimes n} = \sum_{i^n} p_{i^n}\rho_{i^n}$, the sum over all
strings with product weights. Let $T$ be the set of type-typical
strings; $p(T) \ge 1-\epsilon$ for large $n$, and each $i^n \in T$
has letter counts $k_i(i^n) \le n(p_i + \delta')$.

\emph{Step 0: expurgation.} Let $\Pi$ be the typical projector of
$\rho^{\otimes n}$ (I3). Type-typicality of a string does \emph{not}
by itself imply that $\Pi$ accepts its conditional state --- a string
whose type is tilted by $O(\delta')$ can have all its conditional
surprisal outside the global window, so no uniform acceptance
statement over $T$ is available. What (I3) does give is acceptance on
average:
$\sum_{i^n} p_{i^n}\,\mathrm{Tr}(\rho_{i^n}\Pi)
= \mathrm{Tr}(\rho^{\otimes n}\Pi) \ge 1-\epsilon$. Since each term
is at most $1$, Markov's inequality applied to
$1 - \mathrm{Tr}(\rho_{i^n}\Pi)$ gives $p(G) \ge 1-\sqrt\epsilon$ for
the ``globally good'' set
$G = \{\,i^n : \mathrm{Tr}(\rho_{i^n}\Pi) \ge 1-\sqrt\epsilon\,\}$.
Work from now on with the expurgated set
$\widetilde T = T \cap G$: it has
$p(\widetilde T) \ge 1 - \epsilon - \sqrt\epsilon$, retains the
letter-count bounds of $T$, and adds the \emph{uniform} acceptance
$\mathrm{Tr}(\rho_{i^n}\Pi) \ge 1-\sqrt\epsilon$ for every
$i^n \in \widetilde T$. Let
$q_{i^n} = p_{i^n}/p(\widetilde T)$ be the conditional distribution
on $\widetilde T$.

\emph{Step 1: cut-down states.} Let $\Pi_{i^n}$ be the conditional
typical projectors (I3); conditional typicality is a per-string
statement, so $\mathrm{Tr}(\rho_{i^n}\Pi_{i^n}) \ge 1-\epsilon$
uniformly over type-typical strings for large $n$. Define
$\rho'_{i^n} = \Pi\,\Pi_{i^n}\rho_{i^n}\Pi_{i^n}\Pi$ for
$i^n \in \widetilde T$. Two applications of (I4) give uniform
gentleness on $\widetilde T$. First,
$\|\rho_{i^n} - \tau_{i^n}\|_1 \le 2\sqrt\epsilon$ with
$\tau_{i^n} = \Pi_{i^n}\rho_{i^n}\Pi_{i^n}$. Second, the
subnormalized $\tau_{i^n}$ satisfies
$\mathrm{Tr}(\tau_{i^n}\Pi) \ge \mathrm{Tr}(\rho_{i^n}\Pi) -
\|\rho_{i^n}-\tau_{i^n}\|_1 \ge 1 - 3\sqrt\epsilon$, so
$\mathrm{Tr}\,\tau_{i^n} - \mathrm{Tr}(\tau_{i^n}\Pi) \le
3\sqrt\epsilon$ and (I4) in its subnormalized form gives
$\|\tau_{i^n} - \Pi\tau_{i^n}\Pi\|_1 \le 2\sqrt{3\sqrt\epsilon}$.
Hence
\[
\|\rho_{i^n} - \rho'_{i^n}\|_1 \;\le\; \zeta(\epsilon)
:= 2\sqrt\epsilon + 2\sqrt{3\sqrt\epsilon}
\;\xrightarrow[\epsilon\to0]{}\; 0,
\qquad\text{uniformly over } i^n \in \widetilde T .
\]
For the operator norm: (I3) gives
$0 \le \tau_{i^n} \le 2^{-n(\bar S - c\delta')}\,I$ for a constant
$c$ depending only on the ensemble; compressing by $\Pi$ preserves
this bound,
$\rho'_{i^n} = \Pi\tau_{i^n}\Pi \le 2^{-n(\bar S - c\delta')}\,\Pi$,
so $\|\rho'_{i^n}\| \le 2^{-n(\bar S - c\delta')}$. (Note that the
Loewner comparison $\Pi\tau\Pi \preceq \tau$ is false in general;
only the compression of the scalar bound is used.) Consequently
$\rho^{(n)\prime} = \sum_{i^n\in \widetilde T} q_{i^n}\rho'_{i^n}$
satisfies
$\|\rho^{\otimes n} - \rho^{(n)\prime}\|_1 \le \zeta'(\epsilon)
:= \zeta(\epsilon) + 2(\epsilon+\sqrt\epsilon) \to 0$ (the second
term is the mass renormalization from conditioning on
$\widetilde T$), and its support lies in the range of $\Pi$, of
dimension $d \le 2^{n(S(\rho)+\delta')}$.

\emph{Step 2: restriction to large eigenvalues.} Let $\Pi'$ project
onto the eigenvectors of $\rho^{(n)\prime}$ with eigenvalue at least
$\epsilon/\mathrm{Tr}\,\Pi \ge \epsilon\,2^{-n(S(\rho)+\delta')}$.
The discarded weight is at most $\epsilon$, so with
$\sigma_{i^n} = \Pi'\rho'_{i^n}\Pi'$ and
$\tilde\rho = \sum_{i^n\in \widetilde T}q_{i^n}\sigma_{i^n}$ we keep
$\|\rho^{\otimes n} - \tilde\rho\|_1 \le \zeta''(\epsilon) \to 0$
(using (I4) once more), while on its support
$\tilde\rho \ge \epsilon\,2^{-n(S(\rho)+\delta')}\,\Pi'$ and each
$\|\sigma_{i^n}\| \le 2^{-n(\bar S - c\delta')}$.

\emph{Step 3: sampling.} Rescale:
$X_j = \sigma_{I^n(j)}\, 2^{n(\bar S - c\delta')}$, where
$I^n(1),\ldots,I^n(K)$ are drawn i.i.d.\ from $q$; then
$0 \le X_j \le I$ and the mean is
$M = \tilde\rho\,2^{n(\bar S - c\delta')} \ge a\Pi'$ with
$a = \epsilon\,2^{-n(\chi + (1+c)\delta')}$, where
$\chi = S(\rho) - \bar S$. Choose
\[
K = 2^{\,\lceil n(\chi + (3+c)\delta')\rceil} ,
\]
a power of two, so that
$\log_2 K = \lceil n(\chi+(3+c)\delta')\rceil \le
n(\chi+(3+c)\delta') + 1$. Then $K a \ge \epsilon\,2^{2n\delta'}$
(enlarging $K$ only helps the Chernoff bound), so the failure
probability in (I5) is at most
$2d\exp(-\epsilon^3 2^{2n\delta'}/(2\ln2))$, which tends to zero (the
Gaussian-type decay beats the exponential dimension factor
$d \le 2^{n(S+\delta')}$). Hence for large $n$ there \emph{exist}
strings $i^n(1),\ldots,i^n(K) \in \widetilde T$ with
\[
\left\|\frac1K\sum_{j=1}^{K}\sigma_{i^n(j)} - \tilde\rho\right\|_1
\le \epsilon\,\mathrm{Tr}\,\tilde\rho \le \epsilon ,
\]
(trace-norm form of the two-sided operator bound: if
$-\epsilon\tilde\rho \le \Delta \le \epsilon\tilde\rho$ then
$\|\Delta\|_1 \le \epsilon\,\mathrm{Tr}\,\tilde\rho$, by splitting
the trace along the positive part of $\Delta$).

It remains to undo the cut-downs on the \emph{sampled} strings.
Step~1 bounds $\|\rho_{i^n} - \rho'_{i^n}\|_1 \le \zeta$ uniformly on
$\widetilde T$, but the further cut by $\Pi'$ need not be uniformly
gentle for every individual sampled string; it is gentle on average.
Taking traces in the operator sandwich,
$\frac1K\sum_j \mathrm{Tr}\,\sigma_{i^n(j)} \ge
(1-\epsilon)\,\mathrm{Tr}\,\tilde\rho \ge (1-\epsilon)(1-\zeta'')$,
while each $\mathrm{Tr}\,\rho'_{i^n(j)} \le 1$; so the average trace
lost to $\Pi'$, namely $\kappa := \frac1K\sum_j
\bigl(\mathrm{Tr}\,\rho'_{i^n(j)} - \mathrm{Tr}\,\sigma_{i^n(j)}
\bigr)$, is at most $1 - (1-\epsilon)(1-\zeta'')$. By the gentle
operator lemma (I4) applied to each subnormalized $\rho'_{i^n(j)}$
and concavity of the square root,
\[
\frac1K\sum_j\bigl\|\rho'_{i^n(j)} - \sigma_{i^n(j)}\bigr\|_1
\;\le\; \frac1K\sum_j 2\sqrt{\mathrm{Tr}\,\rho'_{i^n(j)} -
\mathrm{Tr}\,\sigma_{i^n(j)}}
\;\le\; 2\sqrt{\kappa} .
\]
Combining the last three displays by the triangle inequality and
convexity of the trace norm,
\begin{equation}
\left\|\frac1K\sum_{j=1}^{K}\rho_{i^n(j)} - \rho^{\otimes n}
\right\|_1 \;\le\; \zeta + 2\sqrt\kappa + \epsilon + \zeta''
\;=:\; \theta(\epsilon) \;\xrightarrow[\epsilon\to0]{}\; 0 .
\tag{A$'$.1}\label{eq:sampled}
\end{equation}

\emph{Step 4: the protocol.} To create $\rho^{\otimes n}$
approximately:
\begin{enumerate}
\item \emph{Shared randomness.} Consume
  $\log_2 K = \lceil n(\chi + (3+c)\delta')\rceil \le
  n(\chi + (3+c)\delta') + 1$ ebits; each party measures their half
  in the computational basis, yielding perfectly correlated uniform
  bits. Because $K$ is a power of two, these bits select a common
  index $j \in \{1,\ldots,K\}$ that is \emph{exactly} uniform --- no
  surplus outcomes and no bias to correct.
\item \emph{Block preparation.} The selected string
  $i^n(j) \in \widetilde T$ contains letter $i$ at most
  $n(p_i+\delta')$ times. For each letter $i$, pick a \emph{fixed}
  regularization block size $m_i$ with
  $E_P(\rho_i^{\otimes m_i})/m_i \le E_P^\infty(\rho_i) + \delta/2$
  (possible since the infimum equals the limit), and let
  $|\psi_i\rangle$ be an optimal purification of
  $\rho_i^{\otimes m_i}$. Use dilution (I2) to share
  $\lceil k_i/m_i\rceil$ copies of $|\psi_i\rangle$ --- this
  \emph{growing} copy count, which plays the role of $m$ in (I2), is
  distinct from the fixed block size $m_i$ --- at ebit rate
  $E_P(\rho_i^{\otimes m_i})/m_i \le E_P^\infty(\rho_i) + \delta/2$
  per copy of $\rho_i$, plus $\delta/2$ dilution slack, with
  $O(\sqrt n)$ communication; discarding the ancillas of each
  $|\psi_i\rangle$ leaves (an approximation to) the block
  $\rho_i^{\otimes k_i}$. Rounding $k_i$ up to a multiple of $m_i$
  wastes at most $m_i - 1 = O(1)$ copies per letter.
\item \emph{Forgetting the index.} Discard the randomness register.
  The output, averaged over $j$, is within
  $\theta(\epsilon) + (\text{dilution errors, vanishing in } n)$ of
  $\rho^{\otimes n}$ in trace norm, by \eqref{eq:sampled} and
  convexity of the trace norm.
\end{enumerate}

\emph{Step 5: rates and limits.} The total ebit consumption is at
most
\[
n(\chi + (3+c)\delta') + 1 \;+\;
\sum_i n(p_i + \delta')\left(E_P^\infty(\rho_i) + \delta\right)
+ O(1),
\]
and the classical communication is $O(\sqrt n) = o(n)$. Dividing by
$n$ and letting $n \to \infty$, then $\epsilon \to 0$, then
$\delta, \delta' \to 0$:
\[
E_{LOq}(\rho) \;\le\; \chi + \sum_i p_i\,E_P^\infty(\rho_i) .
\]
By (I1), $E_P^\infty(\rho) = E_{LOq}(\rho)$, which is the lemma; the
stated equivalent form follows by substituting the definition of
$\chi$ and moving $S(\rho)$ to the left side. \hfill$\blacksquare$

\medskip

\emph{Attribution.} The construction is exactly the one outlined by
Chen and Winter \cite{ChenWinter}; the contribution of this appendix
is only to write every step, error term, and rate out of the
published ingredients (I1)--(I5), so that the main theorem rests on
no unpublished proof.

\renewcommand{\thesection}{B}
\setcounter{subsection}{0}
\section{Certificate semantics and the supplementary artifact}

\begingroup\sloppy
The artifact consists of four files: \texttt{verify\_leaves.py},
\texttt{verify\_upper\_chain.py}, a \texttt{README.md} with run
instructions and expected outputs, and \texttt{MANIFEST.sha256}. Both
programs use only the Python standard library, make every accepting
decision by comparing integers or exact rationals, and contain no
\texttt{assert} statements (so interpreter optimization flags cannot
remove a check). The programs verify the finite calculations; the
paper proves why those calculations apply to every purification.
\par\endgroup

\subsection{Exact arithmetic primitives}

\emph{Outward intervals.} Every quantity over a box is enclosed in an
interval with exact rational endpoints; sums, differences, products,
squares, and cubes propagate endpoints outward (away from the true
range). An interval comparison such as ``$G_{\mathrm{hi}} < 0$'' is a
single exact rational comparison.

\emph{Square roots.} $\sqrt v$ for rational $v \ge 0$ is enclosed by
consecutive dyadic rationals
$\ell = m/2^k \le \sqrt v \le (m{+}1)/2^k = u$ (with $k = 48$ bits in
the leaf verifier's generation pass, $k = 96$ in its audit pass, and
$k = 256$ in the upper-chain verifier), certified by the integer
comparisons $\ell^2 \le v \le u^2$. The upper-chain verifier
additionally handles $\sqrt{33}$ \emph{symbolically}, computing power
sums exactly in $\mathbb{Q}[\sqrt{33}]$; its dyadic brackets enter
only the entropy bound (B.3, item 3).

\emph{Logarithms.} For $r \in [1,2]$, $\ln r$ is enclosed by the
positive series
\[
\ln r = 2\sum_{j=0}^{m-1}\frac{z^{2j+1}}{2j+1} + R_m,
\qquad z = \frac{r-1}{r+1},
\qquad
0 \le R_m \le \frac{2\,z^{2m+1}}{(2m+1)(1-z^2)},
\]
with the exact rational remainder bound added to the upper endpoint
($m = 10$ terms in the leaf verifier's generation pass, $m = 18$ in
its audit pass, and $m = 160$ in the upper-chain verifier). General
arguments are range-reduced by powers of two. A certified lower bound
on $h(v) = v\log_2(1/v)$ for rational $v \in (0,1)$ writes
$1/v = 2^e r$ with $r \in [1,2)$ and uses
$h(v) \ge v\,(e + \ln_{\mathrm{lo}}(r)/\ln_{\mathrm{hi}}(2))$; since
$h$ is concave, its minimum over an interval is attained at an
endpoint, so interval lower bounds need only the two endpoint
evaluations.

\subsection{The leaf verifier}

\texttt{verify\_leaves.py} reconstructs the entire subdivision from
the split rule alone --- no leaf file is read. Starting from the root
box $[1/2,1]\times[0,1/2]\times[0,1/3]$, each box is classified by
the six predicates of Section~\ref{sec:certificate} (in the order
\texttt{qneg}, \texttt{order23}, \texttt{tailcap},
\texttt{orientation}, \texttt{spectral}, \texttt{entropy}, using the
generation-pass precision); an unclassified box is split at the exact
midpoint of its widest coordinate, first coordinate winning ties, up
to a depth guard of $60$. For a box whose tail-weight interval
crosses zero, the range of $Q$ is intersected with $[0,1]$ before
$\sqrt{Q}$ or $h(Q)$ is evaluated --- legitimate because those terms
are only used against spectra with true $Q \in [0,1]$, while the
\texttt{qneg} predicate fires only when $Q_{\mathrm{hi}} < 0$ (the
worked leaf of Section~\ref{sec:certificate} shows this clipping in
action). The program then (i) rechecks every leaf's recorded
predicate at the higher audit precision; (ii) verifies that the leaf
paths are prefix-free, that every internal node has both children,
and that the Kraft sum $\sum 2^{-\mathrm{depth}}$ is exactly $1$ ---
which together prove that the closed leaf boxes cover the root box;
(iii) checks the node count ($50{,}765$), leaf count ($25{,}383$),
maximum depth ($38$), and the class counts of
Section~\ref{sec:certificate}; and (iv) recomputes the SHA-256 hash
of the canonical leaf stream. That stream --- one JSON line per leaf
with its path, depth, predicate class, and exact box --- is
regenerated in memory and bound by the expected hash
\begin{center}
\footnotesize
\texttt{9b06e0fc680ea7ed241e53178297e75b7bce18b7e04f7d9760c5ab14510a26a1}
\end{center}
it is not a separately shipped file. The program prints \texttt{PASS}
and the summary counts only if every check succeeds.

The program also accepts three command-line options that deliberately
corrupt one component of the certificate (a wrong correlation
constant, a wrong entropy threshold, a truncated leaf set); the
verifier rejects each corrupted test case. These runs demonstrate
that the checks are effective, not decorative.

\subsection{The upper-chain verifier}

\texttt{verify\_upper\_chain.py} checks, in exact arithmetic:
\begin{enumerate}
\item the endpoint isometry: the Gram matrix of the four mapped Bell
  columns equals $I_4$;
\item the spectrum: $\mathrm{Tr}\,\rho_{AA'}^k = 3s^k + (1-3s)^k$ for
  $k = 1,\ldots,4$ symbolically over $\mathbb{Q}[\sqrt{33}]$
  (elements as exact rational pairs $a + b\sqrt{33}$; no rounding
  anywhere in the power sums);
\item the first directed certificate $U_w \le U_+$ (equation
  \eqref{eq:Uplus}), using $256$-bit dyadic square-root brackets and
  $160$-term logarithm series (B.1) --- the only place dyadic
  enclosures of $\sqrt{33}$ enter;
\item the exact mixture identity
  $W(1/200) = \tfrac12W(0) + \tfrac12W(1/100)$, entry by entry over
  exact rationals, and the $f{=}0$ identities
  $W(0) = P_{\mathrm{sym}}/3$, $\mathrm{Tr}_B W(0) = I/2$;
\item the second directed certificate: the right side of
  \eqref{eq:chain}, assembled from the three exact entropy
  expressions, is strictly below $T_+$;
\item the exact margin
  \[
  \tfrac{97}{100} - T_+ =
  \frac{3673584061412026000906031263559971238115457039892003881773}
  {10^{60}} > 0 .
  \]
\end{enumerate}

\subsection{Reproducibility}

Running \texttt{python3 verify\_leaves.py} and
\texttt{python3 verify\_upper\_chain.py} from the artifact directory
prints \texttt{PASS} with the summary values quoted above (observed
wall times about $5.6$ and $1.0$ seconds under Python 3.13). The
\texttt{README.md} lists SHA-256 hashes of both programs and of their
complete expected standard outputs, so an independent run can be
compared byte for byte:

\begin{center}
\footnotesize
\begin{tabular}{@{}l@{}}
\toprule
\texttt{verify\_leaves.py}\\
\quad\texttt{2561d18d06d107b9d12648bd458385f704fe6bd2c22e298bdfee0e7f8ef96a29}\\[2pt]
\texttt{verify\_upper\_chain.py}\\
\quad\texttt{e296d111b6e8cff10e195042f313ace3728e37f8575dac79e9c6b3e668581bf9}\\[2pt]
output of \texttt{verify\_leaves.py}\\
\quad\texttt{495dc02cf43d77ef5da2336a86aacc48d62297a579d28205a130564273e21b66}\\[2pt]
output of \texttt{verify\_upper\_chain.py}\\
\quad\texttt{3b07fb96f5467a85f311bef0978b38c474790e41a0c165d7d93c2da8064dc7ca}\\
\bottomrule
\end{tabular}
\end{center}

\end{document}